\documentclass{article} 

\usepackage{amsmath,amssymb,amsthm,natbib,graphicx,url}
\usepackage[a4paper, total={6in, 8in}]{geometry}

\usepackage{booktabs} % For formal tables
\usepackage[ruled]{algorithm2e} % For algorithms

\usepackage[hidelinks]{hyperref}
\usepackage{cleveref}
\usepackage{authblk}

\usepackage{thmtools} 
\usepackage{thm-restate}
\usepackage{MnSymbol} 

\SetAlFnt{\small}
\SetAlCapFnt{\small}
\SetAlCapNameFnt{\small}
\SetAlCapHSkip{0pt}
\IncMargin{-\parindent}
\usepackage{color}

\newcommand{\cM}{\mathcal M}

\newcommand{\cI}{\mathcal I}

\newcommand{\supp}{\mathsf{supp}}
\newcommand{\E}[1]{\mathbb{E}\left[ #1 \right]}
\renewcommand{\P}[1]{\mathbb{P}\left( #1 \right)}

\newcommand{\util}{\mathsf{u}}
\newcommand{\utili}[2]{\mathsf{u}\left({#1}; {#2}\right)}

\newcommand{\pre}{\mathsf{pre}}
\newcommand{\suff}{\mathsf{suff}}

\newcommand{\inst}{\mathcal{I}}
\newcommand{\indicator}[1]{\mathbb{I}_{\{#1\}}}
\newcommand{\transformeps}[1]{\mathcal{T}^{#1}}
\newcommand{\transformber}{\mathcal{T}}

\newcommand{\istar}{{i^*}}
\newcommand{\piStar}{\pi^*}

\newcommand{\A}{r}

\newcommand{\B}[1]{{\color{blue} #1}}
\newcommand{\R}[1]{{\color{red} #1}}

\newtheorem{observation}{Observation}
\newtheorem{claim}{Claim}[section]
\newtheorem{definition}{Definition}[section]
\newtheorem{proposition}{Proposition}[section]

\newtheorem{lemma}{Lemma}[section]
\newtheorem{theorem}{Theorem}[section]

\newtheorem{example}{Example}

\title{Pandora's Problem with Combinatorial Cost\thanks{Michal Feldman and Ben Berger are partially supported by the European Research Council (ERC) under the European Union’s
Horizon 2020 research and innovation program (grant agreement No. 866132), by the Israel Science Foundation
(grant number 317/17), by an Amazon Research Award, and by the NSF-BSF (grant number 2020788).
Federico Fusco and Tomer Ezra are partially supported by ERC Advanced Grant 788893 AMDROMA ``Algorithmic and Mechanism Design Research in Online Markets'' and MIUR PRIN project ALGADIMAR ``Algorithms, Games, and Digital Markets''.}}

\author[1]{Ben Berger}
\author[2]{Tomer Ezra}
\author[1]{Michal Feldman}
\author[2]{Federico Fusco}

\affil[1]{Tel Aviv University, Tel Aviv, Israel}
\affil[2]{Sapienza Universit\`a di Roma, Rome, Italy}

\date{}

\begin{document}

\maketitle

% Abstract. Note that this must come before \maketitle.
\begin{abstract}
Pandora’s problem is a fundamental model in economics that studies optimal search strategies under costly inspection.
In this paper we {initiate} the study of Pandora’s problem with {\em combinatorial costs},
capturing many real-life scenarios where search cost is non-additive. Weitzman's celebrated algorithm [1979] establishes the remarkable result that, for additive costs, 
the optimal search strategy is non-adaptive and computationally feasible. 

We inquire to which extent this structural and computational simplicity extends beyond additive cost functions. Our main result is that the class of submodular cost functions admits an optimal strategy that follows a fixed, non-adaptive order, thus preserving the structural simplicity of additive cost functions. In contrast, for the more general class of subadditive (or even XOS) cost functions the optimal strategy may already need to determine the search order adaptively.
On the computational side, obtaining any approximation to the optimal utility requires super polynomially many 
queries to the cost function,
even for a strict subclass of submodular cost functions. 
\end{abstract}

\clearpage

% Paper body
\section{Introduction}
\label{sec:intro}
Pandora's problem captures the challenge of searching for a good alternative
among multiple options, 
under costly evaluation.
This problem was introduced in the seminal paper of \citet{weitzman79}, as a stochastic search problem 
over $n$ boxes,  each associated with an independent hidden stochastic value, and an exploration cost.
At every point in time, the decision maker chooses which box (if any) to open. 
Upon opening a box, the decision maker incurs its exploration cost, and observes its realized value.
Then, the decision maker can either 
decide 
to open an additional box or halt and obtain the maximum value observed so far.
The goal is to maximize the expected maximum value over the set of opened boxes minus the sum of their exploration costs. 

This setting captures many real-life scenarios, such as hiring employees or searching for an apartment, where there is an inherent tension between the desire to explore many options in an attempt to find one with high reward, and the desire to minimize the total exploration cost.
\citet{weitzman79} 
showed that the optimal strategy 
for this problem
exhibits both structural simplicity and computational simplicity. 
In particular, it opens the boxes according to a fixed-order, determined at the outset; only the stopping time is determined online, depending on the observed values.  Moreover, the entire optimal strategy can be computed efficiently. 

The last few years have seen a renewed interest in Pandora's problem, leading to a line of work that studies
several extensions of the original model.
Most studies focus on extending one of two features of the original problem: either considering a different notion of value derived from the set of opened boxes \cite[e.g.,][]{Olszewski15,Singla18}; or modifying the rules of exploration \cite[e.g.,][]{Doval18,EsfandiariHLM19,BoodaghiansFLL20,FuLX18}.
However, all of them share one fundamental assumption, namely that
each box is associated with an individual cost, and these costs accumulate additively just like in the original model.

However, in many real-life scenarios, exploring one alternative may affect the exploration cost of other alternatives. For instance, when recruiting a new employee, there is a fixed cost for setting up the hiring process, while evaluating each additional candidate induces a small marginal cost.
As another example, when searching for an apartment, each individual visit incurs a cost, but visiting multiple apartments in the same neighborhood is clearly less expensive than the sum of the costs of visiting them separately.

In this paper, we initiate the study of Pandora's problem with combinatorial cost functions, namely, a cost function that assigns a real value to every set of boxes. 
In this model, a decision maker who opens an extra box, 
given a set $S$ of opened boxes, incurs its {\em marginal cost} given $S$. 
We inquire to which extent the structural and computational simplicity of \citet{weitzman79} extends beyond additive cost functions.
As it turns out, the structural simplicity of the original problem does not carry over to general cost functions. 
In particular, the exploration order in the optimal strategy may unavoidably be {\em adaptive}. 
This is demonstrated in the following example.

\begin{example}
\label{ex:adaptive-order}
    Consider an instance with $3$ boxes. The value in box $1$ is $10$ with probability $\tfrac 12$ and $0$ otherwise. The value in box $2$ is $12$ with probability $\tfrac 12$ and $0$ otherwise. The value in box $3$ is $10$ with probability $1$.
    The total cost of exploring a set of boxes from the collection $\{\emptyset, \{1\},\{2\},\{3\},\{1,2\},\{1,3\}\}$ is $0$, and the total cost of exploring a set of boxes from the collection $\{\{2,3\},\{1,2,3\}\}$ is $20$.
    It is not too difficult to observe that opening both boxes 2 and 3 is too expensive for any reasonable strategy. In fact, it can be shown (see Claim~\ref{cl:proof-example}) that the (unique) optimal strategy for this instance is the following: open box $1$. If its value is $10$, then open box $2$, otherwise (\emph{i.e.}, the value in box 1 is $0$), open box $3$. 
\end{example}

In the example above, boxes 2 and 3 exhibit strong complementarity in their cost; namely, the cost of opening both of them is (much) greater than the sum of their individual costs (which is 0). 
Many real-life scenarios, however, exhibit the opposite phenomenon, where the cost of the whole is smaller than the sum of the costs of its parts.  
This structure is captured by the class of subadditive cost functions, where $c(S \cup T) \leq c(S) + c(T)$ for any sets of boxes $S$ and $T$, also known as complement-free functions.

A widely-encountered subclass of subadditive functions is the class of submodular functions, defined by {\em decreasing marginal contribution}. 
Indeed, many real-life exploration tasks exhibit this structure; e.g., where some fixed cost is incurred, followed by smaller individual costs. 
A hierarchy of complement-free functions has been provided by \citet{LehmannLN06}, including
the prominent classes of additive, submodular, and subadditive functions, as well as fractionally-subadditive functions (also known as XOS), where additive $\subset$
submodular $\subset$ XOS $\subset$ subadditive.

Given the prevalence of complement-free cost functions in real-life exploration scenarios, it is natural to study the structure of optimal strategies in these scenarios, and the corresponding computational problem. 
These are the main problems that drive us in this work. 
In particular, we ask whether Pandora's problem 
under different classes of complement-free cost functions preserves the structural and computational simplicity of the original problem with additive costs.

\subsection{Our Results}
\label{sec:results}

{As mentioned above, Example~\ref{ex:adaptive-order} shows an example of a {\em general} cost function, where an adaptive exploration order is inevitable. 
We first show that this phenomenon is not unique to cost functions that exhibit complementarities. 
Indeed, there exist instances with XOS cost functions for which an adaptive exploration order is inevitable\footnote{{Notably, for the larger class of subadditive cost functions, we find that an example demonstrating the necessity of adaptive order can be induced by a (seemingly unrelated) example that has been given in a completely different model of Pandora's box under constrained exploration order~\citep{BoodaghiansFLL20} (see Claim~\ref{cl:subadditive-example} in the Appendix).
}}.}

\vspace{0.1in}

\noindent {\bf Theorem 1 (see \Cref{thm:XOSadaptive}):} There exists an instance of the Pandora's problem with an XOS cost function that admits no optimal strategy with non-adaptive exploration order.

\vspace{0.1in}

{On the face of it, the above theorem seems to be unrelated to Example~\ref{ex:adaptive-order}, where the cost-function exhibits strong complementarity. 
However, we identify a close connection between the two results. 
In particular, we show that every
{instance with a}
(monotone and normalized) cost function over $n$ boxes induces an ``equivalent" 
{instance with an}
XOS cost function over $n+1$ boxes, that inherits the adaptive exploration order of its source
{instance}
(see \Cref{cl:transform-to-xos}).
With this result, the necessity of an adaptive order 
{under}
XOS cost functions can be derived from Example~\ref{ex:adaptive-order}.}

{A key property of XOS functions that enables this construction is that a marginal function of an XOS function
{$c$}
(namely, 
{for some fixed $T$, $c'(S) :=$}
$c(S\mid T) := c(S \cup T)-c(S)$) is unrestricted, and 
{in particular}
can exhibit complementarities.}

{In stark contrast, the class of submodular functions is closed under marginal value; namely, if the cost function $c$ is submodular, then so is the function $c(\cdot \mid T)$ for any fixed set $T$.
In particular, the scenario depicted in Example~\ref{ex:adaptive-order}, where the combined cost of opening boxes 2 and 3 is excessive, while opening each of them separately is cheap, cannot be replicated in an example utilizing a submodular cost function, even with the addition of more boxes.

A natural question is then whether
{instances of}
Pandora's problem with submodular cost functions preserve the structural simplicity of additive costs. 
That is,
{we ask}
whether these instances admit optimal strategies that
{open}
the boxes according to a fixed, non-adaptive order. 
Our first main result answers this question in the affirmative (see Sections \ref{sec:MTT} and \ref{sec:reduction}).}

\vspace{0.1in}

\noindent {\bf Theorem 2 (see \Cref{thm:main}):} Every instance of Pandora's problem with a submodular cost function admits an optimal strategy with non-adaptive exploration order.

\vspace{0.1in}

{Our second main result shows that, while the structural simplicity is preserved under submodular cost functions, the computational simplicity is not preserved. 
In particular, in \Cref{sec:computational} we prove the following stronger result.}

\vspace{0.1in}
\noindent {\bf Theorem 3 (see \Cref{thm:oracle}):} The problem of deciding whether a given instance of Pandora's problem with a submodular cost function admits a strategy that attains strictly positive utility requires super-polynomially many queries to the cost function.
\vspace{0.1in}

{Notably, this theorem implies that no approximation to the optimal utility can be obtained with polynomially-many cost queries.}

\subsection{Our Techniques}
\label{sec:techniques}

    The main technical tool to solve Pandora's problem is the notion of reservation value of a box \cite[e.g.,][]{weitzman79,KleinbergK18,BoodaghiansFLL20,EsfandiariHLM19,Singla18}. 
    This is the maximum value, presumably among those observed in previously opened boxes, for which opening the box achieves the same marginal utility as not opening it. 
    Formally, the reservation value of a box with random reward $V$ and (additive) cost $c$ is the solution $z$ of the following equation: $\E{(V-z)^+} = c$. Weitzman's optimal strategy opens the boxes in decreasing order of reservation value, halting when the current maximum observed reward exceeds the reservation value of the next
    {unopened}
    box. 
    Since they are also easy to compute, reservation values simultaneously establish structural and computational simplicity for the problem.
    In the combinatorial setting that we study, however, this approach may yield an arbitrarily bad performance. 
    \begin{example}
        Consider an instance with $2$ identical boxes, each with a random reward of $2$ with probability $\frac{1}{3}$ (and $0$ otherwise), and a symmetric unit-demand cost function with a cost of $1$ (i.e., $c(\{1\})=c(\{2\})=c(\{1,2\})=1$, and $c(\emptyset)=0$). The reservation value of the two boxes is negative, thus Weitzman's strategy would not open any one of them. However, the best strategy for this instance opens both boxes, achieving an expected utility of $2\cdot \frac{5}{9}-1 >0$.
    \end{example}

    The example illustrates why the reservation value is not suitable in the presence of combinatorial costs: the intrinsic importance of a box in the exploration is not solely determined by its random reward or its current marginal cost, but also by its influence on the marginal cost of all the (exponentially many) possible subsets of boxes that may be opened in the future. 
    
    In what follows we describe our techniques for our structural and computation results.
    {We first present our techniques for the main structural result for Bernoulli instances, and then show how to extend it from Bernoulli to general instances. 
    Finally, we present our techniques for our computational impossibility result.}

    \vspace{0.1in}
    \noindent \textbf{Bernoulli instances.}
    In \Cref{sec:MTT} we prove Theorem 2 for Bernoulli instances, \emph{i.e.}, instances 
    where each box $i$ has value $v_i$ with probability $p_i$ and value 0 otherwise.

    A key notion in our analysis is that of an \emph{impulsive strategy}.  
    Such a strategy is determined by an ordered subset of boxes, 
    and proceeds by opening them in the given order and halting upon the first time that a non-zero value of a box is observed (or if all boxes of the strategy have been opened).
    We show that every Bernoulli instance admits an optimal strategy that takes the form of an impulsive strategy. To establish this result, we follow the following steps.

We first show that we may assume the existence of an optimal strategy $\piStar$ that takes the following form: 
It starts by opening 
an arbitrary
box $\A$. If its non-zero value is realized, then it executes some impulsive sub-strategy $\pi^Y$, and if its realized value is 0, then it executes another impulsive sub-strategy $\pi^N$.
This is proved by induction, using the fact that the marignal cost of a submodular function is also submodular.

Under this assumption, we proceed as follows:
    {Assume towards contradiction that there is no optimal strategy which is impulsive.}
    If all boxes of $\pi^N$ appear also in $\pi^Y$, then it is straightforward to argue that the impulsive strategy that first executes $\pi^Y$, and then opens $\A$ if no non-zero value was observed, is an impulsive strategy that yields at least the same utility as $\piStar$, {and we are done.} 

    Therefore it remains to handle the case where there exists a box in $\pi^N$ that does not appear in $\pi^Y$.
    In this case, we show that there exists a subset of $\pi^N \setminus \pi^Y$
    that can be concatenated to $\pi^Y$ to improve the overall utility and thus obtain a contradiction.

    The main tool we use to this end is the notion of an \emph{impulsive strategy with dummies}.
    This is a randomized strategy which is determined by a (deterministic) impulsive strategy $\pi$ and a subset $A$ {of its boxes,}
    {denoted $\pi_A$,}
{and proceeds as follows: For a box $i \in A$, it proceeds as usual (open the box, observe its value, incur its marginal cost and halt if the observed value is non-zero).}
{For a box $i \notin A$, instead of opening $i$, it halts with probability $p_i$ and otherwise continues to the next box.
In particular $\pi_A$ only opens boxes from $A$.}

    Such a strategy is appealing, since it 
    {restricts the set of boxes that might be opened}
    while retaining {some of the} properties of the original strategy.
    For example, {the contribution of any box $i\in A$ to the expected reward is the same in $\pi_A$ as in $\pi$.}
    Furthermore,
    {every impulsive strategy with dummies is a probability distribution over deterministic impulsive strategies. Thus, any lower bound on its utility applies also to the utility of the best impulsive strategy in its support.}

    We use the notion of impulsive strategies with dummies to identify a strategy that can be concatenated to $\pi^Y$ which has a positive marginal utility, thus reaching a contradiction.
    In particular, we prove that given an impulsive strategy $\pi$ and any partition $A\cupdot B$ of its boxes,
    the utility attained by $\pi$ is at most the utility of $\pi_B$ 
    plus the \emph{marginal} utility of $\pi_A$ when executed after the boxes in $B$ have been opened.
    The submodularity of the cost function is crucial to obtain this technical property.
    The desired strategy that can be concatenated to $\pi^Y$ can now be identified, by applying this lemma with $\pi := \pi^N$, $A = \pi^N \setminus \pi^Y$, $B = \pi^Y \cap \pi^N$.
    In particular, we prove that there exists such a strategy in the support of $\pi^N_A$.

\vspace{0.1in}
\noindent \textbf{From Bernoulli to arbitrary instances.}
In \Cref{sec:reduction}, we show how to extend Theorem 2 to hold for arbitrary distributions. We do so using the following steps: 
     We devise a transformation that, given an arbitrary instance $\inst$ creates a Bernoulli instance $\inst'$, which maintains 
     %the class of the cost function, 
     {submodularity of the cost function}
     as well as other properties.
    First, the transformation discretizes the (possibly) continuous and unbounded distributions to have finite supports, and then it ``Bernoullifies'' each box by associating it with a set of Bernoulli boxes. 
 
    We {then} show a correspondence between strategies for the two instances in which an impulsive strategy for $\inst'$ is associated with a fixed order strategy for $\inst$. The correspondence preserves the utility up to an {arbitrarily} small precision.
We conclude that if there is an instance that admits a gap between the best fixed-order strategy and the best arbitrary strategy, then it implies that there is a Bernoulli instance that admits a gap between the best impulsive strategy and best arbitrary strategy, contradicting the main result of Section~\ref{sec:MTT}.
   The instance-transformation we use might be of independent interest and find applications in other stochastic settings (such as prophet setting).

    \vspace{0.1in}
    \noindent \textbf{Computational hardness.}
    In \Cref{sec:computational} we prove Theorem 3
    even for a very simple subclass of submodular functions (\emph{i.e}, matroid rank functions).
    To this end, we
    {follow the construction of \citet{SvitkinaF11}, and}
    design two instances of Pandora's box problem whose cost functions are ``indistinguishable''
    using polynomially many cost queries, 
    but only one of them admits a strategy that yields positive utility.
    Since no algorithm can distinguish between them efficiently,
    we conclude that the problem of deciding whether a given instance admits a strategy that attains positive utility is unsolvable with polynomially many 
    cost queries. 
    Moreover, this implies that no approximation can be obtained by an efficient algorithm.

\subsection{Related Work}
\label{sec:related}

    Pandora's Problem originated in economics but has suscitated a keen interest in the computer science community. Weitzman's optimal solution is based on the clever idea of reservation value, a quantity that captures the intrinsic value of a box in the exploration process. The reservation value has a deep connection with the notion of Gittins index \citep{Weber92}; actually, \citet{Dumitriu03} showed that it is possible to rephrase Pandora's problem as a Markov game whose Gittins index coincides with the reservation value. Recently, a simpler proof of the optimality of Weitzman's rule was also given by \citet{KleinbergWW16}. Following these papers, many interesting modifications of Pandora's Problem have been considered.
    
    \citet{Singla18} used an adaptivity gap approach to approximately solve Pandora's problem under various combinatorial models, while \citet{Olszewski15} studied to which extent a threshold strategy like Weitzman's is optimal when the definition of the reward of the exploration goes beyond the $\max$ function. 

    A successful line of work has also focused on Pandora's Problem with non-obligatory inspection. Here, at the end of the exploration, the decision maker can decide to select an  unopened box without having to open it (and thus without paying its cost). \citet{Doval18} introduced this model, highlighting the surprising property that there are instances where the optimal strategy is adaptive in the order of boxes it chooses. A sequence of papers then closed this problem from the computational perspective \citep{BeyhaghiK19,BeyhaghiC22,FuJD22}: Pandora's problem with non-obligatory inspection is NP-hard to solve but a PTAS exists for it. Interestingly enough, this minor tweak in the exploration rule (i.e., giving the possibility of getting a single box ``for free'' without inspection) hindered both the computational and structural simplicity of the original setting.

    Constraints on the order in which the boxes can be opened have also been studied. \citet{EsfandiariHLM19} considered the case where the boxes have to be opened consistently with a total ordering of the boxes (possibly skipping some). In contrast, \citet{BoodaghiansFLL20} investigated partial orderings on the boxes modeled by precedence graphs. In that work, the authors investigated to which extent the simplicity of the original Pandora's problem extends under order constraints: when the partial ordering on the boxes is represented by a tree, then there exists an optimal strategy that is fixed order and can be computed efficiently; however, under general partial ordering the problem becomes NP-hard to solve, and there are instances where adaptivity is needed to achieve optimality. We further elaborate on the relations with our work in \Cref{app:pandora}.

    \citet{FuLX18,Segev021} studied Pandora's Problem with commitment, when, similarly to what happens in online selection problems like secretary or prophet inequalities, only the reward in the last opened box can be collected. \citet{ChawlaGTTZ20,ChawlaGMT21} investigated what happens when the assumption on the independence of the random rewards in the boxes is dropped, \citet{AlaeiMM21} introduced the revenue maximization version of the problem, while \citet{BechtelDP22} considered a delegated version of Pandora's problem.
    Finally, Pandora's problem has also been studied from the learning perspective, both in the sample complexity framework \citep{Guo0T021}, and in online learning \citep{GergatsouliT22,Gatmiry22}.

\section{Preliminaries}
\label{sec:preliminaries}
In Pandora's problem there are $n$ boxes, containing hidden values $V_i$ which are distributed according to the independent non-negative distributions $D_i$.
We
denote by $\supp$ the union of the supports of these distributions.
The cost of inspecting a set of boxes is given by a combinatorial cost function $c:2^{\left[n\right]} \rightarrow \mathbb{R}_{\geq 0}$, where $[n]$ denotes the set $\{1,\ldots,n\}$.  We assume that $c$ is always normalized and monotone, \emph{i.e}, $c(\emptyset) = 0$ and $S\subseteq T$ implies $c(S) \leq c(T)$. We also use $(x)^+$  to denote $\max(x,0)$ for any number $x \in \mathbb{R}$.

We denote an \emph{instance} of the problem by $\inst = (D_1,\ldots,D_n,c)$.
Given an instance $\inst$, a \emph{strategy} $\pi$ for $\inst$ inspects the boxes in a sequential manner where each inspection of box $i$ reveals its hidden (random) value $V_i$.
At each round the strategy may choose any uninspected box to inspect next, or it may halt and attain as utility the difference between the largest observed value and the cost of the set of opened boxes. The decisions are based on the given instance and the sequence of opened boxes and realized values so far.
Given a strategy $\pi$ for $\inst$, we use the following notation:
\begin{itemize}
    \item $S(\pi)$ - the (random) ordered set of boxes opened by $\pi$.
    \item $V(\pi) := \max_{i \in S(\pi)} V_i$ - the maximum value observed by $\pi$.  We also refer to this as the reward obtained by $\pi$.
    \item $\util(\pi) := \E{V(\pi)} - \E{c(S(\pi))}$ - the expected utility (i.e., value minus cost) achieved by $\pi$.
\end{itemize}
Note that the quantities defined above depend on the given instance. When not clear from the context, we shall use $\utili{\inst}{\pi}$ to denote the expected  utility of strategy $\pi$ for instance $\inst$. 

A \emph{randomized} strategy can toss coins before every decision point.  Note that these coins can be tossed a priori before the first box is inspected.  Therefore
every randomized strategy is a distribution over deterministic strategies.  In particular, for every randomized strategy there is a deterministic strategy that achieves at least the same utility (the one with the highest utility in the support of the distribution).

Given an instance $\inst$, we denote by $\Pi$ the set of all strategies for $\inst$.
An \emph{optimal} strategy $\pi^*$ is a strategy that maximizes the utility, \emph{i.e.}, $\pi^* = \arg\max_{\pi \in \Pi} \util(\pi)$.
By the paragraph above, we can assume without loss of generality that $\pi^*$ is deterministic. Note also that if there is some $i$ for which $\E{V_i} = \infty$, then the strategy that opens $i$ and halts achieves infinite utility.
We thus assume that all distributions $D_i$ have finite expectations.

A \emph{fixed order strategy} $\pi$ is a strategy in which the order of inspection is non-adaptive.  Formally, such a strategy is characterized by a
permutation $\sigma: \left[n\right] \rightarrow \left[n\right]$ such that at every round $i$, the strategy either opens the box $\sigma(i)$, or halts.
A strategy $\pi$ is called a \emph{fixed order strategy with thresholds} $t_1,\ldots,t_{n} \in \mathbb{R}$ if it is fixed order, and at every round $i$, $\pi$ halts if and only if the maximum value inspected so far is at least $t_i$.
The proof of the following observation is deferred to \Cref{app:preliminaries}.

\begin{restatable}{observation}{OptimalFixThresholds}\label{obs:fix-thresholds}
    For every permutation $\sigma$, the optimal strategy with fixed order $\sigma$ is a fixed order strategy with thresholds.
\end{restatable}

A \emph{Bernoulli instance} is an instance where all distributions $D_i$ are weighted Bernoulli distributions,
\emph{e.g.}, with probability 0.7 $V_i = 18$ and otherwise $V_i = 0$.
An \emph{impulsive} strategy for a Bernoulli instance is a fixed order strategy that immediately halts if the value of the currently inspected box is non-zero (the strategy can also halt if the currently observed value is zero).  
An example of such a strategy is:  inspect box 1 and halt if its value is non-zero.  Otherwise, inspect box 2 and halt if its value is non-zero.  Otherwise, inspect box 7 and halt (regardless of the findings). An example of a non-impulsive strategy is:  inspect box 1.  If its value is non-zero, inspect box 2 and halt.  Otherwise, inspect box 3 and halt.
Note that an impulsive strategy is a fixed order strategy, where each threshold equals the weight of its corresponding Bernoulli box
(except for the threshold corresponding to the last box, which equals 0).
We also remark that the empty strategy which halts immediately without inspecting any boxes is considered an impulsive strategy.

\paragraph{Combinatorial functions.} In this paper we study combinatorial cost functions. In particular, given a base set $X$
of elements, we say that a function $c:2^X \to \mathbb{R}_{\geq0}$ is
\begin{itemize}
    \item submodular if $c(x\mid B) \le c(x \mid A)$ for all $A \subseteq B\subseteq X$, $x \in X\setminus B$, where $c(x \mid S):= c(S \cup \{x\}) - c(S)$ denotes the marginal contribution of element $x$ to set $S.$ 
    \item fractionally subadditive (XOS) if there exists a family of linear function $\{c_i\}$ such that $c(A) = \max_i c_i(A)$, for all $A \subseteq X$.
    \item subadditive if $c(A \cup B) \le c(A) + c(B)$ for all $A,B \subseteq X$.    
\end{itemize}
It is known that submodular $\subset$ XOS $\subset$ subadditive, with strict inclusions \citep{LehmannLN06}.

\paragraph{Computational setting.}
The computational problem we consider is the following (Pandora's) decision problem:
given an instance $\inst$, decide whether there exists a strategy $\pi$ for $\inst$ that achieves positive utility, \emph{i.e}, $\utili{\inst}{\pi} > 0$.
An algorithm for this problem gets access to the given cost function via cost queries (analogous to value queries for a combinatorial valuation function); namely, given a set $S$ of elements, a cost query returns $c(S)$.

\section{Impulsive Optimal Strategies for Bernoulli Instances}
\label{sec:MTT}

In this section, we prove our main structural result for the special case of Bernoulli instances. 
\begin{theorem}
\label{thm:MTT_Submodular}
For every Bernoulli instance with a submodular cost function there exists an optimal strategy that is impulsive. 
\end{theorem}
The crux of the proof of Theorem \ref{thm:MTT_Submodular} is 
captured by the following Lemma, which is the main technical result of the paper.
\begin{lemma}
\label{lem:one_fork_wlog}
Let $\inst$ be a Bernoulli instance with a submodular cost function.
If there exists an optimal strategy for $\inst$ of the following form: 
\begin{itemize}
    \item Inspect some first box, denoted $\A$, that follows the distribution $V_\A = v_r > 0$ with probability $p_\A > 0$.
    \item If  $V_\A = v_\A$,
            execute an impulsive sub-strategy $\pi^Y$.
    \item If  $V_\A = 0$,
            execute an impulsive sub-strategy $\pi^N$. 
\end{itemize}
Then, there exists an optimal strategy for $\inst$ which is impulsive.
\end{lemma}

Before proving Lemma~\ref{lem:one_fork_wlog}, we show how it implies Theorem~\ref{thm:MTT_Submodular}.

\begin{proof}[Proof of Theorem \ref{thm:MTT_Submodular}]
We prove this by induction on the number of boxes $n$.
For $n=1$, the claim is trivially true, since every strategy is an impulsive strategy.
Assume by induction that for any Bernoulli instance $\inst'$ on $n-1$  boxes with a submodular cost function there exists a deterministic optimal strategy that is impulsive.
Let $\inst = \left(D_1,\ldots,D_n,c\right)$ be a Bernoulli instance with $n$ boxes whose cost function is submodular, and let $\piStar$ be a deterministic optimal strategy for $\inst$. 
Since the strategy $\piStar$ is deterministic, it either does not open any box (and thus $\piStar$ is an impulsive strategy), or there exists a box $i\in [n]$ that it inspects first. 
Note that if $V_i = 0$ with probability 1,
then $\pi'$can be weakly improved by skipping $i$ and proceeding to the next box:
the value obtained by this new strategy is the same for any realization of the boxes, but the incurred cost is weakly improved (by monotonicity of the cost function).
Thus we can assume without loss of generality that $V_i = v_i > 0$ with some probability $p_i > 0$, and $V_i = 0$ otherwise.

For each of the two possible realizations of box $i$, the instance remaining after opening box $i$ is either $\inst^N=(D_1,\ldots ,D_{i-1},{D_{i+1},\ldots},D_n,c')$ if $V_i=0$, or $\inst^Y=(D'_1,\ldots ,D'_{i-1},{D'_{i+1},\ldots},D'_n,c')$ if $V_i=v_i$, where $c':[n]\setminus \{i\}\rightarrow \mathbb{R}_{\geq 0}$ is the cost function $c'(S) = c(S \cup\{i\})-c(\{i\}) = c(S \mid \{i\})$, and $D'_j$ is the weighted Bernoulli distribution of $(v_j-v_i)^+$ with probability $p_j$ where $D_j$ is the Bernoulli distribution of having a value of $v_j$ with probability $p_j$.
Note that since $c'$ is the marginal function of $c$ given $\{i\}$ and since $c$ is submodular, then $c'$ is a submodular function.
By the induction hypothesis (since $\inst^Y,\inst^N$ have submodular cost functions and $n-1$ boxes), 
there exist two optimal strategies $\pi^Y$, $\pi^N$ for $\inst^Y, \inst^N$, respectively, that are impulsive.
Thus there exists an optimal strategy that opens box $i$, if its value is non-zero executes the sub-strategy $\pi^Y$, and otherwise it execute the sub-strategy $\pi^N$. By applying  Lemma~\ref{lem:one_fork_wlog}, we establish that there exists an optimal strategy for $\inst$ that is impulsive.
\end{proof}

The remainder of this section is dedicated to the proof of Lemma \ref{lem:one_fork_wlog}. In Section~\ref{sec:setup} we make the required preparation, and in Section~\ref{sec:proof-body} we provide the full proof of the lemma.

\subsection{Setup for Lemma \ref{lem:one_fork_wlog}}
\label{sec:setup}
In this section we introduce the notation and constructs that we shall need for the proof of Lemma \ref{lem:one_fork_wlog}.
Let $\inst = (D_\A, D_1,\ldots,D_n,c)$ be a Bernoulli instance where $c$ is a submodular cost function.
For every $i \in \{\A\}\cup[n]$, the random value $V_i$ in box $i$ is set to $v_i$ with probability $p_i$, and to $0$ otherwise (with probability $q_i = 1 - p_i$).
We can assume without loss of generality that $v_i > 0$ and $p_i > 0$ for every box $i \in \{\A\} \cup [n]$,
since otherwise $V_i = 0$ with probability 1, 
in which case 
%opening box $i$ can only decrease the utility of a strategy --- 
any strategy that does open $i$ can be weakly improved by skipping $i$ and proceeding as if its value 0 was observed.
If a box $i$ satisfies $p_i = 1$ then we say it is a \emph{deterministic} box.

An impulsive (sub-)strategy $\pi$ is given by a tuple of box indices (with no repetitions), \emph{e.g}, $(1,2,7)$ stands for the impulsive strategy that first inspects box 1 and halts if $V_1 = v_1$, otherwise it proceeds to inspect box 2 and halts if $V_2 = v_2$, 
%\tec{$  V_2$?}, 
and otherwise it proceeds to inspect box 7 and halts.
An impulsive strategy can also be given by a tuple of impulsive sub-strategies $(\pi_1,\ldots,\pi_k)$, \emph{e.g.}, $((1,2),(7))$ stands for the strategy $(1,2,7)$.
The empty strategy that does not inspect any box is also considered an impulsive strategy and is denoted by the tuple $(\emptyset)$.
We shall occasionally abuse notation and identify an impulsive strategy $\pi$ with the set of boxes that form $\pi$,
\emph{e.g.}, $i \in (1,2,7)$ stands for $i \in \{1,2,7\}$, and $\pi \subseteq \{1,2,3,4\}$ means that all boxes outside of $\{1,2,3,4\}$
%have 0 probability of being inspected by $\pi$.
are never inspected by $\pi$.
%\mfc{Would it be better to write that boxes outside of $\{1,2,3,4\}$ are never inspected by $\pi$?}
%\bbc{Yes.}

Let $\piStar$ be a deterministic optimal strategy for $\inst$ in the form given by the statement of Lemma \ref{lem:one_fork_wlog}.
Thus, $\piStar$ first inspects box $\A$; if it observes that $V_\A = v_\A$ then it executes the impulsive sub-strategy $\pi^Y$ and otherwise it executes the impulsive sub-strategy $\pi^N$.  Note that $\pi^Y$ and $\pi^N$ both inspect boxes with indices from $[n]$.
We can assume without loss of generality that for every $i \in \pi^Y$, we have $v_i \geq v_\A$: 
%as otherwise opening box $i$ in the scenario that $V_\A = v_\A$ 
%cannot increase the maximum observed value, and therefore can 
Otherwise $\piStar$ can be weakly improved by 
%replacing the inspection of box $i$ with a decision to halt with probability $p_i$ and otherwise (with probability $q_i$) continue to the next box.
removing $i$ from $\pi^Y$ --- note that the reward obtained in the end of the process is unaffected by the realized value of $V_i$ in this case, and therefore continuing to the suffix of $\pi^Y$ after $i$ is also optimal.
%two scenarios where $V_i = v_i$ or $V_i =0$ do not impact the reward obtained in the end of the process.}
%\bbc{Was this clear?} \tec{I think Yes}
% 
We also assume that each of $\pi^Y$ and $\pi^N$ contains at most one deterministic box, in which case it is the last one in the tuple.
This too is without loss of generality since impulsive strategies always halt after inspecting a deterministic box.
%\bbc{Consider moving the above sentence to the part where we describe impulsive strategies in general.}
Note that if $\pi^Y$ is the empty strategy (\emph{i.e} it halts immediately without opening any boxes), then $\piStar$ is an impulsive strategy by itself, and we are done.
Thus we assume that $\pi^Y$ is not empty, \emph{i.e.}, $\left|\pi^Y\right| \geq 1$.
Finally, out of all optimal strategies that satisfy the assumptions above, we also assume that $\piStar$ maximizes $\left|\pi^Y\right| + \left|\pi^N\right|$.

Assume towards contradiction that there is no impulsive strategy for $\inst$ that achieves the same utility as $\piStar$.  
We show that in this case we can replace either $\pi^Y$ or $\pi^N$ by impulsive sub-strategies of bigger size, without losing utility.
This would constitute a contradiction to the definition of $\piStar$. Given an impulsive strategy $\pi$, we denote by $p_{\left(\pi\right)}$ the probability that one of the boxes inspected by $\pi$ has a non-zero value, 
%\tec{I think we should replace realized throghout the paper with non-zero. This term is confusing},
\emph{i.e.}, the probability that there is some $i\in S(\pi)$ such that $V_i = v_i$.
We denote by $q_{\left(\pi\right)} := 1 - p_{\left(\pi\right)}$ the probability that 
%none of the inspected boxes had their values realized.
$V_i = 0$ for every $i\in \pi$.
For the empty strategy we define $p_{(\emptyset)}=0$ (or equivalently $q_{(\emptyset)}=1$).
Note that by our assumption that $p_i > 0$ for every $i$, we have $p_{\left(\pi\right)} > 0$ for every non-empty impulsive strategy, and $p_{\left(\pi\right)} = 1$ if and only if $\pi$ contains a deterministic box.
\begin{observation}
\label{obs:p_q}
%\bbc{Change later to also encompass concatenation of strategies.}
    Let $\pi = \left(i_1,\ldots,i_k\right) \subseteq [n]$ be an impulsive strategy. Then: 
    \begin{itemize}
        \item $p_{\left(\pi\right)} = \sum_{j=1}^k q_{\left(i_1,\ldots,i_{j-1}\right)}\cdot p_{i_j}$.
        \item $q_{\left(\pi\right)} = \prod_{j=1}^k q_{i_j} = \prod_{j=1}^k (1- p_{i_j})$.
    \end{itemize}
\end{observation}
Note that Observation \ref{obs:p_q} also holds when the coordinates $i_j$ are by themselves impulsive sub-strategies which are not singletons.
Also observe that if $\pi^1,\pi^2$ are impulsive strategies such that $\pi^1 \subseteq \pi^2$, then $p_{\left(\pi^1\right)} \leq p_{\left(\pi^2\right)}$.

We now introduce notation for the marginal utility achieved by an impulsive (sub) strategy executed at some point after 
inspecting box $\A$.
%the first step of the execution of $\pi^*$. 
%the scenario where $V_\A = v_\A$ and the scenario where $V_\A = 0$. 
Note that this quantity depends on whether the observed value $V_\A$ equals $v_\A$ or 0.
%executing an impulsive (sub)-strategy after inspecting $\A$ and observing that $V_\A = v_\A$ 
%does not achieve the same utility as executing it after inspecting $\A$ and observing that $V_\A = 0$.
We thus introduce notation for both cases, and it shall be useful to define these utilities conditioned on already 
%paying for some set of boxes $T$ (which presumably were already inspected).
having inspected some set of boxes $T$.
We also introduce a third ``non-lower-bounded utility'' that we shall need.
%\tec{ removed "in the third bullet."} 
\begin{definition}
\label{def:marginal_utilities_bernoulli}
Given an impulsive strategy $\pi \subseteq [n]$ and a set of boxes $T \subseteq [n]$ such that $T\cap \pi = \emptyset$, we define
%\mfc{As before, I suggest changing below to $Y$ and $N$.}
%\bbc{As before, I agree.}
\begin{itemize}
    \item $\util_Y(\pi \mid T) := \E{\max_{i \in S(\pi)} (V_i - v_\A)^+} - \E{c\left(S(\pi) \mid \{\A\} \cup T\right)}$, the marginal utility of $\pi$, given that $V_\A = v_\A$ and that the boxes in $T$ were already opened.
    \item $\util_N(\pi \mid T) := \E{\max_{i \in S(\pi)} (V_i)} - \E{c\left(S(\pi) \mid \{\A\} \cup T\right)}$, the marginal utility of $\pi$, given that $V_\A = 0$ and that the boxes in $T$ were already opened.
    \item %\bbe{
    $\util_M(\pi \mid T) := p_{\left(\pi\right)}\cdot \E{\max_{i \in S(\pi)} (V_i - v_\A) \mid \exists i \in S(\pi) \textrm{ s.t. } V_i = v_i}
    - \E{c\left(S(\pi) \mid \{\A\} \cup T\right)}$.
    %}
    %\tec{why not write it the same way as the former ones?}
\end{itemize}
\end{definition}

We write $\util_Y(\pi), \util_N(\pi), \util_M(\pi)$ instead of $\util_Y(\pi \mid \emptyset), \util_N(\pi \mid \emptyset), \util_M(\pi \mid \emptyset)$, respectively.
We observe that since $c$ is submodular, then for any sets of boxes $T_1 \subseteq T_2$ that do not intersect $\pi$, we have
$\util_Y(\pi \mid T_1) \leq \util_Y(\pi \mid T_2), \util_N(\pi \mid T_1) \leq \util_N(\pi \mid T_2)$ and $\util_M(\pi \mid T_1) \leq \util_M(\pi \mid T_2)$ .

\begin{observation}
\label{obs:utility_expansion}
Let $\pi = (i_1,\ldots,i_k) \subseteq [n]$ be an impulsive strategy. Then:
%, where $i_j \in [n]$ for every $j$.  Then we have:
\begin{align*}
    \util_N(\pi) &= 
    \sum_{j=1}^k q_{\left(i_1,\ldots,i_{j-1}\right)}\cdot \util_N(i_j \mid \{i_\ell\}_{\ell \in [j-1]}) 
    =\sum_{j=1}^k q_{\left(i_1,\ldots,i_{j-1}\right)}\cdot \left(p_{i_j}\cdot v_{i_j} - c\left(i_j \mid \{\A\}\cup \{i_\ell\}_{\ell \in [j-1]}\right)\right). \\
    \util_Y(\pi) &=
    \sum_{j=1}^k q_{\left(i_1,\ldots,i_{j-1}\right)}\cdot \util_Y(i_j \mid \{i_\ell\}_{\ell \in [j-1]}) 
    =\sum_{j=1}^k q_{\left(i_1,\ldots,i_{j-1}\right)}\cdot \left(p_{i_j}\cdot \left(v_{i_j} - v_\A\right)^+ - c\left(i_j \mid \{\A\} \cup \{i_\ell\}_{\ell \in [j-1]}\right)\right). \\
    \util_M(\pi) &=
    \sum_{j=1}^k q_{\left(i_1,\ldots,i_{j-1}\right)}\cdot \util_M(i_j \mid \{i_\ell\}_{\ell \in [j-1]})
    =\sum_{j=1}^k q_{\left(i_1,\ldots,i_{j-1}\right)}\cdot \left(p_{i_j}\cdot \left(v_{i_j} - v_\A\right) - c\left(i_j \mid \{\A\} \cup \{i_\ell\}_{\ell \in [j-1]}\right)\right).
\end{align*}
\end{observation}
The corresponding expressions $\util_N(\pi\mid T),\util_Y(\pi \mid T),\util_M(\pi \mid T)$ for a set $T\subseteq [n]$ such that $T\cap \pi = \emptyset$ follow the same equations above with the addition of a ``$T$'' term after every ``$\mid$'' symbol.
% \begin{proof}
% \bbc{Should I add a proof for this?} \tec{It is ok not to prove this}
% \end{proof}
As a concrete example, for the strategy $\pi = (1,2,7)$ and set of boxes $T= \{4,5\}$, we have
$$
\util_N(\pi \mid T) = p_1v_1 - c(1 \mid \{\A,4,5\}) + q_{1}\left(p_2v_2 - c(2 \mid \{\A,4,5,1\})\right) + q_{(1,2)}\left(p_7v_7 - c(7\mid \{\A,4,5,1,2\})\right).
$$
Furthermore, observe that
$$
\util(\piStar) =  p_\A \cdot (v_\A + \util_Y(\pi^Y)) - c(\A) + q_\A \cdot \util_N(\pi^N).
$$
%\bbc{Change the following to something correct.}
Our goal is to replace either $\pi^Y$ or $\pi^N$ with a strategy $\pi$ that achieves at least as much marginal utility, but (potentially) inspects more boxes.
This will constitute a contradiction to the assumption that $\pi^*$ maximizes $\left|\pi^Y\right| + \left|\pi^N\right|$.
%find a strategy $\pi \subseteq [n]$ such that $\util_Y(\pi) \geq \util_Y(\pi^Y)$ and $\left|\pi\right| \geq \left|\pi^Y\right|$, or $\util_N(\pi) \geq \util_N(\pi^N)$ and.

The proof of the following straightforward observation is deferred to Appendix~\ref{app:MTT}.
\begin{restatable}{observation}{bottomEqualsTopMinusPpi}
\label{obs:bottom_equals_top_minus_p_pi_v_pi}
    Let $\pi \subseteq [n]$ be an impulsive strategy.  Then for any set of boxes $T\subseteq [n]$ such that $T\cap\pi = \emptyset$, we have:
    \begin{itemize}
        \item $\util_M(\pi \mid T) \leq \util_Y(\pi \mid T) \leq \util_N(\pi \mid T)$.
        \item $\util_M(\pi \mid T) = \util_N(\pi \mid T) - p_{\left(\pi\right)}\cdot v_\A$.
        \item If $\pi \subseteq \pi^Y$, then $\util_M(\pi \mid T) = \util_Y(\pi \mid T)$.
    \end{itemize}
\end{restatable}

\paragraph{Impulsive Strategies with Dummies.}
 Our proof makes use of a particular family of 
 %randomized
 strategies that are distributions over impulsive strategies:
 an \emph{impulsive strategy with dummies} is given by a (regular) impulsive strategy $\pi$, and a subset of boxes $P \subseteq \pi$.
 The strategy is denoted $\pi_P$, and proceeds exactly as $\pi$ would, with the following single difference:
 when considering index $i\in \pi$, if it is also the case that $i \notin P$ (\emph{i.e.}, $i \in \pi \setminus P$), then instead of inspecting box $i$ the strategy rather only halts with probability $p_i$ and otherwise proceeds to the next coordinate of the tuple.  We refer to the boxes in $\pi \setminus P$ as \emph{dummy} boxes.
 %\bbc{Fix example to so that $T$ contains a dummy and a non-dummy box, and explain that this is allowed. Change $\pi$ to $(1,2,7,4)$}
 As an example, the strategy $(2,1,4,7)_{\{1,7\}}$ first halts with probability $p_2$, then, if it did not halt it proceeds to inspect box 1 and halts if $V_1 = v_1$, otherwise it halts with probability $p_4$, and then, if it did not halt it proceeds to inspect box 7 and halts.
 % inspects box 1 and halts if $V_1 = v_1$, otherwise it halts with probability $p_2$, if it did not halt it proceeds to inspect box 7 and halts if $V_7 = v_7$, and otherwise it inspects box 4 and halts.
 Observe that such a strategy is a distribution over deterministic impulsive strategies.
 %For example,  $(1,2,7,4)_{\{1,7\}}$  equals the strategy $(1)$ with probability $p_2$, the strategy $(1,7)$ with probability $q_2\cdot p_7$, and the strategy $(1$.
 For example, $(2,1,4,7)_{\{1,7\}}$ equals the empty strategy with probability $p_2$, the strategy $(1)$ with probability $q_2\cdot p_4$, and the strategy $(1,7)$ with probability $q_{(2,4)} = q_2\cdot q_4$.
 The marginal utility quantities in Definition \ref{def:marginal_utilities_bernoulli} carry over to impulsive strategies with dummies.
 For example, given the strategy $\pi = (2,1,4,7)$, subset of boxes $P = \{1,7\}$ and another set of boxes $T= \{4,5\}$ that has presumably already been opened, we have
$$
%\util_N(\pi_P \mid T) = p_1v_1 - c(1 \mid \{\A,4,5\}) + q_{(1,2)}\left(p_7v_7 - c(7\mid \{1\} \cup \{\A,4,5\})\right).
\util_M(\pi_P \mid T) = q_2\left(p_1(v_1-v_r) - c(1 \mid \{\A,4,5\})\right) + q_{(2,1,4)}\left(p_7(v_7-v_r) - c(7\mid \{1\} \cup \{\A,4,5\})\right).
$$
Note that the value and cost terms corresponding to boxes 1 and 7 are multiplied by the factors $q_2$ and $q_{(2,1,4)}$, respectively, and that these are 
the same factors these terms are multiplied by in the expression for the utility of $\pi$.
Also note that $T\cap \pi \neq \emptyset$ in this example, but we allow this since the strategy we are computing the utility for, $\pi_P$, never inspects boxes from $T$.
Furthermore, the expression $p_{\left(\pi_P\right)}$ ---  the probability that one of the boxes inspected by $\pi_P$ has a non-zero value --- is also well defined.
\emph{E.g}, in the example above this probability equals $q_2p_1 + q_{(2,1,4)}p_7$.

\Cref{obs:bottom_equals_top_minus_p_pi_v_pi} also carries over to impulsive strategies with dummies, where the third bullet there holds for any such strategy $\pi_P$ where $P \subseteq \pi^Y.$
Finally, observe that for $P=\emptyset$ we have that $\pi_P$ is the empty strategy, and that for $P=\pi$ we have that 
%$\pi_P = \pi$.
$\pi_P$ coincides with $\pi$.

\subsection{Proof of Lemma \ref{lem:one_fork_wlog}}
\label{sec:proof-body}

The first step of the proof of Lemma \ref{lem:one_fork_wlog}  is the following inequality.

% \mfe{We first show that $p_{\left(\pi^Y\right)} < p_{\left(\pi^N\right)}$.}

\begin{lemma}
\label{lem:easy_cases}
%We must have
It holds that $p_{\left(\pi^Y\right)} < p_{\left(\pi^N\right)}$.
    % The following conditions must hold:
    % \begin{enumerate}
    %     % \item $\exists i \in \pi^Y \setminus \pi^N$.
    %     % \item $\exists i \in \pi^N \setminus \pi^Y$.
    %     %\item $\pi^Y \subsetneq \pi^N$.
    %     \item $p_{\left(\pi^Y\right)} < 1$, \emph{i.e.}, the last box in $\pi^Y$ is not a deterministic box.
    %     \item $\pi^Y$ is not a permutation of $\pi^N$.
    % \end{enumerate}
    % \tec{Why not $p_{\left(\pi^Y\right)} < p_{\left(\pi^N\right)}$?} 
    % \bbc{Otherwise the only case left to handle is $\exists i \in \pi^N \setminus \pi^Y$.
    % But I also want to handle the case whose solution is to concatenate boxes from the upper branch to the lower branch since the power of Lemma \ref{lem:dummy_split_lemma} is more apparent there.}
\end{lemma}

\begin{proof}
    Assume towards contradiction that 
    %the lemma does not hold.
    %Thus, either $p_{\left(\pi^Y\right)} = 1$,
    %or $\pi^Y$ is a permutation of $\pi^N$ in which case we have $p_{\left(\pi^Y\right)} = p_{\left(\pi^N\right)}$.
    %In any of these cases we have 
    $p_{\left(\pi^Y\right)} \geq p_{\left(\pi^N\right)}$.  This implies
    $$
    \util_N\left(\pi^Y\right) = \util_Y\left(\pi^Y\right) + p_{\left(\pi^Y\right)}v_\A \geq \util_Y\left(\pi^N\right) + p_{\left(\pi^N\right)}v_\A
    \geq \util_N\left(\pi^N\right) \geq \util_N\left(\pi^Y\right),
    $$
    where the equality and the second inequality hold by Observation \ref{obs:bottom_equals_top_minus_p_pi_v_pi}, 
    the first inequality holds by the optimality of $\pi^Y$ for the scenario that $V_\A = v_\A$,
    and the last inequality holds by the optimality of $\pi^N$ for the scenario that $V_\A = 0$.
    
    Thus all expressions in the above chain are equal and in particular we have $\util_N\left(\pi^Y\right) = \util_N\left(\pi^N\right)$.
    This implies that the strategy $\pi'$ that first inspects $\A$ and then executes $\pi^Y$ regardless of the realization of $V_\A$ is also optimal.
    Now consider the impulsive strategy $\pi'' = \left(\pi^Y, \A\right)$.
    Since $v_i \geq v_\A$ for any $i \in \pi^Y$,
    %\bbc{Check this assumption}
    then the maximum value observed by $\pi''$ coincides with that of $\pi'$ for any realization of the boxes.
    On the other hand the cost incurred by $\pi''$ is weakly less then that of $\pi'$, again for any realization of the boxes.
    Thus the impulsive strategy $\pi''$ is optimal as well, a contradiction.
\end{proof}

The following lemma is the main technical tool needed for the rest of the proof.
% The following lemma is the main technical tool used for the proof of Lemma \ref{lem:one_fork_wlog}.
%we use to prove Theorem \ref{thm:MTT_Submodular}.

\begin{lemma}
\label{lem:dummy_split_lemma}
    Let $\pi\subseteq[n]$ be any impulsive strategy, and let $\pi = A \cupdot B$ be a partition of the set of boxes corresponding to $\pi$.  Then we have
    % \begin{enumerate}
    %     \item $\util_N(\pi) \leq \util_N(\pi_A \mid B) + \util_N(\pi_B)$.
    %     \item $\util_Y(\pi) \leq \util_Y(\pi_A \mid B) + \util_Y(\pi_B)$.
    % \end{enumerate}
    $\util_N(\pi) \leq \util_N(\pi_A \mid B) + \util_N(\pi_B)$.
\end{lemma}
%\bbc{Consider moving this lemma and is proof to the setup section.}\tec{I think here it is better}
\begin{proof} 
%\bbc{Ready for feedback}
Let $\pi,A,B$ be as in the lemma statement.
Recall that the expressions in the 
%inequalities of the lemma statement
inequality 
are each made up of (expected) value terms and (expected) cost terms.
We first show that the value terms cancel out.
%each of the inequalities.  
%We show this for part (1) of the lemma, but the analogous claim for part (2) holds if in the following paragraph we replace every term $v_i$ or $V_i$ by $(v_i-v_\A)^+$ or $(V_i-v_\A)^+$, respectively.
Explicitly, we show that $\E{\max_{i \in S(\pi)} (V_i)} = \E{\max_{i \in S(\pi_A)} (V_i)} + \E{\max_{i \in S(\pi_B)} (V_i)}$.

To see this, denote $\pi$ without loss of generality as $\pi = (1,\ldots,k)$.
Then, for any $i \in [k]$, the value term corresponding to $i$ when expanding $\E{\max_{i \in S(\pi)} (V_i)}$ is $q_{(1,.\ldots,i-1)}p_iv_i$.  
Furthermore, regardless of whether $i \in A$ or $i \in B$, this would also be the value term corresponding to $i$ when expanding the right-hand side of the equation --- if $i \in A$ then this would appear in the expansion of $\E{\max_{i \in S(\pi_A)} (V_i)}$ and if $i \in B$ then this would appear in the expansion of $\E{\max_{i \in S(\pi_B)} (V_i)}$.
In fact, this last discussion also shows:
\begin{observation}
\label{obs:sum_of_p's}
    For any impulsive strategy $\pi \subseteq [n]$ and for any partition $\pi = A \cupdot B$ of the set of boxes corresponding to $\pi$, we have
    $p_{(\pi)} = p_{\left(\pi_A\right)} + p_{\left(\pi_B\right)}$.  
    %\tec{not consistent when you use parenthesis in $p_\pi$ vs $p_{(\pi)}$}
\end{observation}

It remains to handle the cost terms.
%, and note that these are the same for both parts of the lemma.
For ease of exposition we omit the ``$\{\A\}$'' terms inside the conditional cost terms.
This has no effect on the proof.
Thus, in the remainder of the proof we establish the following inequality:
\begin{align}
    \E{c\left(S(\pi_A) \mid B\right)} + \E{c\left(S(\pi_B)\right)} - \E{c\left(S(\pi)\right)} \leq 0. \label{eq:cost_inequalities} 
\end{align}

We prove inequality (\ref{eq:cost_inequalities}) by induction on $\left|A\right| + \left|B\right|$,
% and we start with the base case $\left|A\right| + \left|B\right| = 1$.
% If $\left|A\right| =1, \left|B\right| = 0$, then in particular $\pi_A = \pi$ and $\pi_B = \emptyset$, implying that
% $\E{c\left(S(\pi_A) \mid B\right)} = \E{c\left(S(\pi)\right)}$ and $\E{c\left(S(\pi_B)\right)} = 0$ and the inequality follows.
% The case $\left|A\right| =0, \left|B\right| = 1$ is handled analogously.
and we start with the base case $\left|A\right| + \left|B\right| = 0$.
In this case, $\pi$ is the empty strategy implying that all three summands in inequality (\ref{eq:cost_inequalities}) equal 0, and the inequality follows.

We now assume that $\left|A\right| + \left|B\right| \geq 1$.
Denote $\pi$ again without loss of generality as $\pi = (1,\ldots,k)$, where $k = \left|A\right| + \left|B\right|$.
We expand each of the expressions in inequality (\ref{eq:cost_inequalities}):
\begin{align*}
    \E{c\left(S(\pi)\right)}    &= \sum_{i=1}^k q_{(1,\ldots,i-1)} \cdot c\left(i \mid \{1,\ldots,i-1\}\right) \\
    \E{c\left(S(\pi_B)\right)} &= \sum_{i \in B} q_{(1,\ldots,i-1)} \cdot c\left(i \mid \{1,\ldots,i-1\} \cap B \right) \\
    \E{c\left(S(\pi_A) \mid B\right)} &= \sum_{i \in A} q_{(1,\ldots,i-1)} \cdot c\left(i \mid \{1,\ldots,i-1\} \cup B \right)
\end{align*}

By plugging these into inequality (\ref{eq:cost_inequalities}) and taking out common ``$q$'' factors, we get the equivalent inequality
\begin{align}
    &\sum_{i \in A} q_{(1,\ldots,i-1)} \cdot 
        \left[c\left(i \mid \{1,\ldots,i-1\} \cup B \right) - c\left(i \mid \{1,\ldots,i-1\}\right) \right]  \label{eq:A_summands}\\
    &\;+\sum_{i \in B} q_{(1,\ldots,i-1)} \cdot 
        \left[c\left(i \mid \{1,\ldots,i-1\} \cap B \right) - c\left(i \mid \{1,\ldots,i-1\}\right) \right]  \label{eq:B_summands} \\
    &\leq 0 \nonumber
\end{align}

We now split to two cases.  In the first (easy) case we assume that $k \in A$, \emph{i.e.}, the last box potentially to be inspected by $\pi$ is a box from $A$.
Note that in this case we have $\{1,\ldots,k-1\} \cup B = \{1,\ldots,k-1\}$,
and thus the summand in line (\ref{eq:A_summands}) corresponding to $i=k$ cancels out and equals 0.
Therefore, if we denote $\pi^{(k)} := (1,\ldots,k-1)$, then inequality (\ref{eq:cost_inequalities}) is equivalent to
$$
\E{c\left(S\left(\pi^{(k)}_{A\setminus\{k\}}\right) \mid B\right)} 
    + \E{c\left(S\left(\pi^{(k)}_B\right)\right)} - \E{c\left(S\left(\pi^{(k)}\right)\right)} \leq 0,
$$
which holds by the induction hypothesis.

We now handle the case $k \in B$.
Note that if $A = \emptyset$ and $B = [k]$, then inequality (\ref{eq:cost_inequalities}) holds trivially --- the summands in line (\ref{eq:A_summands}) do not exist, and the summands in line (\ref{eq:B_summands}) cancel out. We thus assume that $\left|A\right| \geq 1$.

The rest of the proof involves a systematic manipulation of the inequality.  
Mostly, we shall make repeated use of the following observation, which we term the ``cancellation lemma''.
%and whose simple proof we defer to the appendix.
We use colors in the lemma statement so that it will be easier to see how we apply it in the rest of the proof.

\begin{restatable}{lemma}{CancellationLemma} \emph{(Cancellation Lemma)}
\label{lem:cancellation_lemma}
    For every cost function $c:2^{\left[n\right]} \rightarrow \mathbb{R}_{\geq 0}$, subset $T \subseteq [n]$ and elements $\B{h},\R{\ell} \in [n] \setminus T$, we have
    $$
    c\left(\B{h} \mid T \cup \{\R{\ell}\}\right) - c\left(\R{\ell} \mid T \cup \{\B{h}\}\right) = c\left(\B{h} \mid T \right) - c\left(\R{\ell} \mid T \right).
    $$
\end{restatable}

\begin{proof}%[Proof of the Cancellation Lemma]
%\bbc{add colors to the proof}
The lemma holds since
    \begin{align*}
        c&\left(\B{h} \mid T \cup \{\R{\ell}\}\right) - c\left(\R{\ell} \mid T \cup \{\B{h}\}\right) \\
        &=\left[c\left(\{\B{h}\} \cup T \cup \{\R{\ell}\}\right) - c\left(T \cup \{\R{\ell}\}\right) \right]
            -\left[ c\left(\{\R{\ell}\} \cup T \cup \{\B{h}\}\right) - c\left(T \cup \{\B{h}\}\right) \right]\\
        &=c\left(T \cup \{\B{h}\}\right) - c\left(T \cup \{\R{\ell}\}\right) \\
        &=\left[c\left(T \cup \{\B{h}\}\right) - c(T)\right] - \left[c\left(T \cup \{\R{\ell}\}\right) - c(T) \right] \\
        &=c\left(\B{h} \mid T \right) - c\left(\R{\ell} \mid T \right). \qedhere
    \end{align*}
\end{proof}

Denote $A = \{a_1,\ldots,a_w\}$, where $w = \left|A\right| \geq 1$
%\tec{I would change $\alpha$ to another symbol. it is difficult to distinguish between $\alpha $ and a}, 
and where the order $a_1,\ldots,a_w$ is consistent with the relative ordering of $A$ in $\pi$, \emph{i.e.}, $a_1 < \cdots < a_w$.
Thus we can rewrite line (\ref{eq:A_summands}) as follows:
\begin{align}
    \sum_{i=1}^w q_{(1,\ldots,a_i-1)} \cdot 
        \left[c\left(a_i \mid \{a_1,\ldots,a_{i-1}\} \cupdot 
        %\left(B\setminus\{k\}\right)\cup \{k\}
        B
        \right) - c\left(a_i \mid \{1,\ldots,a_i-1\}\right) \right] \label{eq:A_summands_2}
\end{align}
For each $i=1,\ldots,w$ we shall refer to the corresponding summand in the above sum as the ``$a_i$-summand''. 
We can also rewrite the summand in line (\ref{eq:B_summands}) corresponding to $i=k$ (recall that we are in the case that $k \in B$) as
\begin{align*}
q_{(1,\ldots,k-1)} \cdot 
        \left[c\left(k \mid B \setminus \{k\} \right) - c\left(\R{k} \mid \left(B \setminus \{k\} \right) 
        \cupdot \{a_1,\ldots, a_{w-1}, \B{a_w}\}\right) \right]. %\label{eq:k_summand_1}
\end{align*}
and we shall refer to it as the ``$k$-summand''.
Note the coloring of $k$ and $a_w$, highlighting their roles in the (first upcoming) application of the cancellation lemma.
Consider the 
%summand in line (\ref{eq:A_summands_2}) corresponding to $i=\alpha$:
$a_w$-summand:
\begin{align*}
q_{(1,\ldots,a_w - 1)} \cdot 
        \left[c\left(\B{a_w} \mid \{a_1,\ldots,a_{w-1}\} \cupdot \left(B\setminus\{k\}\right) \cupdot \{\R{k}\} \right) - c\left(a_w \mid \{1,\ldots,a_w-1\}\right) \right]
\end{align*}

%\bbc{In the following, write $D$ explicitly}
We cannot directly apply the lemma on the $k$-summand and the $a_w$-summand because of the different ``$q$'' factors.  To get around this issue,
denote the difference inside the square parentheses in the $k$-summand by 
$$
D := \left[c\left(k \mid B \setminus \{k\} \right) - c\left(\R{k} \mid \left(B \setminus \{k\} \right) 
        \cupdot \{a_1,\ldots, a_{w-1}, \B{a_w}\}\right)\right],
$$
and note that $D \geq 0$ since $c$ is submodular.
Furthermore, note that $q_{(1,\ldots,a_w - 1)} \geq q_{(1,\ldots,k-1)}$.
Thus we can (weakly) 
increase 
%manipulate
the $k$-summand as follows:
\begin{align*}
q_{(1,\ldots,k-1)} \cdot D 
%&= q_{(1,\ldots,a_\alpha - 1)} \cdot D 
%+ \left(q_{(1,\ldots,k-1)} - q_{(1,\ldots,a_\alpha - 1)}\right)\cdot D \\
&\leq q_{(1,\ldots,a_w - 1)} \cdot D \\
&= q_{(1,\ldots,a_w - 1)} \cdot
    \left[c\left(k \mid B \setminus \{k\} \right) - c\left(\R{k} \mid \left(B \setminus \{k\} \right) 
        \cupdot \{a_1,\ldots, a_{w-1}, \B{a_w}\}\right) \right]
\end{align*}
%\tec{the first step is redundant}
Now the ``$q$'' factors are the same and we can apply the cancellation lemma.  
Thus we remove the $\{\R{k}\}$ term from the $a_w$-summand, which now becomes
 \begin{align*}
q_{(1,\ldots,a_w - 1)} \cdot 
        \left[c\left(a_w \mid \{a_1,\ldots,a_{w-1}\} \cupdot \left(B\setminus\{k\}\right) \right) - c\left(a_w \mid \{1,\ldots,a_w-1\}\right) \right]  
\end{align*}
%Thus we remove the $\{\R{k}\}$ term from the $a_\alpha$-summand.
We also remove the $\B{a_w}$ term from the $k$-summand, which now becomes
\begin{align*}
    q_{(1,\ldots,a_w-1)} \cdot 
        \left[c\left(k \mid B \setminus \{k\} \right) - c\left(\R{k} \mid \left(B \setminus \{k\} \right) 
        \cupdot \{a_1,\ldots, a_{w-2},\B{a_{w-1}}\}\right) \right].
\end{align*}
Note the coloring of $k$ and $a_{w - 1}$ highlighting the next application of the cancellation lemma.
Consider now the $a_{w - 1}$-summand:
\begin{align*}
q_{(1,\ldots,a_{w-1} - 1)} \cdot
        \left[c\left(\B{a_{w-1}} \mid \{a_1,\ldots,a_{w-2}\} \cupdot \left(B\setminus\{k\}\right) \cupdot \{\R{k}\} \right) - c\left(a_{w - 1} \mid \{1,\ldots,a_{w-1} - 1\}\right) \right]
\end{align*}
As before, we cannot directly apply the cancellation lemma due to the different ``q'' factors.
As before, we get around this by using the fact that $q_{(1,\ldots,a_{w-1}-1)} \geq q_{(1,\ldots,a_w-1)}$ and the fact that $c$ is submodular in order to replace the $q_{(1,\ldots,a_w-1)}$ factor in the $k$-summand by the factor $q_{(1,\ldots,a_{w-1}-1)}$, making the $k$-summand (weakly) larger by doing so.

After the application of the cancellation lemma, we remove the $\{\R{k}\}$ term from the $a_{w-1}$-summand.
We also remove the $\B{a_{w-1}}$ term from the $k$-summand, which becomes
\begin{align*}
    q_{(1,\ldots,a_{w-1}-1)} \cdot 
        \left[c\left(k \mid B \setminus \{k\} \right) - c\left(\R{k} \mid \left(B \setminus \{k\} \right) 
        \cupdot \{a_1,\ldots, a_{w-3},\B{a_{w-2}}\}\right) \right],
\end{align*}
and again note the coloring of $k$ and $a_{w - 2}$ highlighting the next application of the cancellation lemma.
We continue this way, applying the cancellation lemma to the summands corresponding to the pairs
$\left(\R{k},\B{a_{w-2}}\right), \left(\R{k},\B{a_{w-3}}\right),\ldots, \left(\R{k},\B{a_1}\right)$.

After the last application, the $k$-summand becomes
\begin{align*}
    q_{(1,\ldots,a_1-1)} \cdot 
        \left[c\left(k \mid B \setminus \{k\} \right) - c\left(k \mid B \setminus \{k\} 
        \right) \right] = 0,
\end{align*}
and the sum of the $a_i$-summands (Line (\ref{eq:A_summands_2})) is modified by replacing ``$B$'' with ``$B\setminus \{k\}$''.
Thus, recalling the notation
$\pi^{(k)} := (1,\ldots,k-1)$,
we have shown in the above process that
\begin{align*}
&\E{c\left(S(\pi_A) \mid B\right)} + \E{c\left(S(\pi_B)\right)} -\E{c\left(S(\pi)\right)} \leq \\
&\E{c\left(S\left(\pi^{(k)}_{A}\right) \mid B\setminus \{k\}\right)} + \E{c\left(S\left(\pi^{(k)}_{B\setminus \{k\}}\right)\right)} - \E{c\left(S\left(\pi^{(k)}\right)\right)}, 
\end{align*}
and the bottom expression is upper-bounded by 0, by the induction hypothesis.  This concludes the proof of Lemma \ref{lem:dummy_split_lemma}.
\end{proof}

    The remainder of the proof of Lemma~\ref{lem:one_fork_wlog} proceeds as follows. % We are ready to directly address the statement of Lemma \ref{lem:one_fork_wlog}.
%Theorem \ref{thm:MTT_Submodular}. 
By Lemma \ref{lem:easy_cases} we have
%we can assume for the rest of the proof that $p_{\left(\pi^Y\right)} < 1$,
%and we have two remaining cases to handle:
%(i) $\pi^N \subset \pi^Y$ 
%\tec{$ \pi^Y \setminus \pi^N \neq \emptyset$?} \tec{Now I see that I misunderstood this symbol, we need to replace this symbol}
%(i) $\exists i \in \pi^Y \setminus \pi^N$, 
%where the inclusion is strict,
$\pi^N \setminus \pi^Y \neq \emptyset$,
since otherwise $\pi^N \subseteq \pi^Y$ which implies $p_{\left(\pi^N\right)} \leq p_{\left(\pi^Y\right)}$.
Lemma \ref{lem:easy_cases} also implies that $p_{\left(\pi^Y\right)} < 1$ ,\emph{i.e.}, $\pi^Y$ does not contain a deterministic box.
% 
% We start with the more straightforward case (i).
% First, observe that 
% %$\pi^N \subsetneq \pi^Y$
% the assumption of this case
% %\tec{if we fix the former leamma then this case is impossible} 
% implies $p_{\left(\pi^N\right)} \leq p_{\left(\pi^Y\right)} < 1$,
% and in particular $\pi^N$ does not have a deterministic box.
To prove \Cref{lem:one_fork_wlog} we show that there exists a non-empty impulsive sub-strategy made from boxes in $\pi^N \setminus \pi^Y$ that we can concatenate to $\pi^Y$ 
without decreasing utility. 
% while improving utility.
%This way we will obtain a new optimal strategy that constitutes 
This would constitute
a contradiction to the definition of $\piStar$.
Let $A$ and $B$ be the sets of boxes defined by $A = \pi^N \setminus \pi^Y$, $B = \pi^Y \cap \pi^N \subseteq \pi^Y$.
Note that $\pi^N = \ A \cupdot B$ and that 
$A \neq \emptyset$.
%both $A$ and $B$ are non-empty.
 We can write $\pi^N$ as a concatenation of contiguous sub-strategies made up of boxes from $A$ or $B$ as follows:  $\pi^N = \left(B^\pre,A^1,B^1,\ldots,A^k,B^k,A^\suff\right)$, where $A =  \left(\cupdot_{i=1}^k A^k\right) \cupdot A^\suff , B = B^\pre \cupdot \left(\cupdot_{i=1}^k B^k\right)$.
The only sub-strategies that we allow to be empty in this presentation are $B^\pre$, for the case that $\pi^N$ starts with a box from $A$, and $A^\suff$, for the case that $\pi^N$ ends with a box from $B$
(in the latter case we must have $k\geq1$ as otherwise $A = \emptyset$ and we get a contradiction).
%\tec{maybe say something about the case that B is empty. it implies that $k=0$, and then Asuff cannot be empty.}
%\bbc{Move some of the below to the appendix.}

We define the strategies $\pi^A = \left(A^1,A^2,\ldots,A^k,A^\suff\right)$, $\pi^B = \left(B^\pre,B^1,B^2,\ldots,B^k\right)$, and note the difference between $\pi^A,\pi^B$, and $\pi^N_A,\pi^N_B$.  The former are deterministic strategies, whereas the latter are strategies with dummies.
%Note also that $\pi^B$ is a permutation of $\pi^Y$. 

\begin{claim}
\label{clm:u_M(piA mid piY) < 0}
We have  $\util_M\left(\pi^A \mid \pi^Y\right) < 0$.
\end{claim}
\begin{proof}
Assume towards contradiction that $\util_M\left(\pi^A \mid \pi^Y\right) \geq 0$. Then by Observation \ref{obs:bottom_equals_top_minus_p_pi_v_pi}
%We now split to two sub-cases. In the first case we assume that $\util_M\left(\pi^A \mid \pi^Y\right) \geq 0$.  In this case we are done since by Observation \ref{obs:bottom_equals_top_minus_p_pi_v_pi} 
we also have $\util_Y\left(\pi^A \mid \pi^Y\right) \geq 0$. Thus we can
%--- similarly to what we did in case (i)\bbc{fix this explanation} --- 
concatenate $\pi^A$ to $\pi^Y$ 
to obtain a new optimal strategy that contradicts the definition of $\piStar$ as the maximizer of $\left|\pi^Y\right| + \left|\pi^N\right|$.
Formally, consider the strategy obtained from $\piStar$ by replacing the strategy $\pi^Y$ with $\left(\pi^Y,\pi^A\right)$.  Then the utility obtained does not decrease, since
$$
\util_Y\left(\left(\pi^Y,\pi^A\right)\right) = 
    \util_Y\left(\pi^Y\right) + q_{\left(\pi^Y\right)} \util_Y\left(\pi^A \mid \pi^Y\right)
    \geq 
    \util_Y\left(\pi^Y\right).
$$
Thus, the new strategy is optimal as well, and as discussed above we get a contradiction.
\end{proof}

Observe that
\begin{align*}
\util_N\left[\left(\pi^Y,\pi^A\right)\right] \leq \util_N\left(\pi^N\right) \leq \util_N(\pi^N_A \mid B) + \util_N(\pi^N_B)
\leq \util_N(\pi^N_A \mid \pi^Y) + \util_N(\pi^N_B),
\end{align*}
%\bbc{The first inequality can be made strong, otherwise $(\pi^Y,\A,\pi^A)$ is also optimal.We can do this if we want to assume in the setup section that $v_i > v_\A$ for every $i\in \pi^Y$.}
where the first inequality holds by the optimality of $\pi^N$ for the scenario where $V_\A = 0$, 
the second inequality holds by Lemma \ref{lem:dummy_split_lemma},
and the third holds by submodularity of the cost function $c$ since $B \subseteq \pi^Y$.
Now, since the strategy $\left(\pi^Y,\pi^A\right)$ is a superset of $\pi^N$, then in particular we have 
$$
p_{\left(\pi^Y,\pi^A\right)} \geq p_{\left(\pi^N\right)} = p_{\left(\pi^N_A\right)} + p_{\left(\pi^N_B\right)},
$$
where the second equality holds by Observation \ref{obs:sum_of_p's}.
Therefore, the chain of inequalities above implies:
\begin{align*}
\util_M\left[\left(\pi^Y,\pi^A\right)\right] &= \util_N\left[\left(\pi^Y,\pi^A\right)\right] - p_{\left(\pi^Y,\pi^A\right)} v_\A \\
            &\leq \util_N(\pi^N_A \mid \pi^Y) + \util_N(\pi^N_B) - \left(p_{\left(\pi^N_A\right)} + p_{\left(\pi^N_B\right)}\right) v_\A \\
            &= \left(\util_N(\pi^N_A \mid \pi^Y) - p_{\left(\pi^N_A\right)} v_\A \right) + \left(\util_N(\pi^N_B) - p_{\left(\pi^N_B\right)} v_\A \right) \\
            &= \util_M(\pi^N_A \mid \pi^Y) + \util_M(\pi^N_B),
\end{align*} 
%\tec{the equalities don't make sense. you are removing probabilites not multiplied by a value. I guess this is $v_r$}
where the first and last equalities hold by Observation \ref{obs:bottom_equals_top_minus_p_pi_v_pi}.
Since $\util_M\left[\left(\pi^Y,\pi^A\right)\right] = \util_M\left(\pi^Y\right) + q_{\left(\pi^Y\right)}\util_M\left(\pi^A \mid \pi^Y\right)$, then the above inequality implies
\begin{align}
    \util_M\left(\pi^Y\right) - \util_M(\pi^N_B) \leq \util_M(\pi^N_A \mid \pi^Y) - q_{\left(\pi^Y\right)}\util_M\left(\pi^A \mid \pi^Y\right). \label{eq:low_branch_inequality}
\end{align}

In the following claim we rule out the case that $k=0$, \emph{i.e}, that in $\pi^N$ all boxes from $B$ are inspected before all boxes from $A$.  The proof is deferred to Appendix~\ref{app:MTT}.
\begin{restatable}{claim}{kGreaterThanOne}
\label{clm:k>=1}
There exist boxes $a \in A, b \in B$ such that $\pi^N$ inspects $b$ only after inspecting $A$,
\emph{i.e.}, $k \geq 1$.
\end{restatable}
%The proof of \Cref{clm:k>=1} is deferred to Appendix~\ref{app:MTT}.

In the remainder we show that for some $j\in[k]$, we can concatenate the (non-empty) strategy $\left(A^1,\ldots,A^j\right)$ to $\pi^Y$ without losing utility.
%As in case (i),
This would constitute a contradiction to the assumption that $\piStar$ maximizes $\left|\pi^Y\right|+ \left|\pi^N\right|$. To this end we analyze both sides of inequality (\ref{eq:low_branch_inequality}).
First, the left hand side %of inequality (\ref{eq:low_branch_inequality})
satisfies 
\begin{align}
0 \leq \util_Y\left(\pi^Y\right) - \util_Y(\pi^N_B) =  \util_M\left(\pi^Y\right) - \util_M(\pi^N_B) \label{eq:lhss}
\end{align}
where the equality holds by \Cref{obs:bottom_equals_top_minus_p_pi_v_pi} (recall that $B \subseteq \pi^Y$), and the inequality holds since $\pi^Y$ is the optimal sub-strategy for the scenario where $V_\A = v_\A$. For the right hand side we have the following claim which is derived through a careful algebraic manipulation 
%of both sides of inequality (\ref{eq:low_branch_inequality}) 
that mostly applies Observation \ref{obs:utility_expansion}. 
%Its proof is also deferred to Appendix~\ref{app:MTT}.

\begin{restatable}{claim}{lowBranch}
\label{clm:low_branch_inequality_manipulation}
The right hand side
%The two sides 
of inequality (\ref{eq:low_branch_inequality}) 
%satisfy:
satisfies
    \begin{align*}
        %\util_M\left(\pi^Y\right) - \util_M(\pi^N_B) &\geq 
        %    \sum_{\ell=1}^{k} q_{\left(A^1,\ldots,A^{\ell-1}\right)}p_{(A^\ell)}
        %    \left[\util_M\left(\pi^Y\right) - \util_M\left(\left(B^\pre,B^1,\ldots,B^{\ell-1}\right)\right) \right] \\
        \util_M(\pi^N_A \mid \pi^Y) - q_{\left(\pi^Y\right)}\util_M\left(\pi^A \mid \pi^Y\right) &\leq 
            \sum_{j=1}^k  q_{\left(B^\pre, B^1,\ldots,B^{j-1}\right)} p_{\left(B^j\right)}  
            \util_M\left(\left(A^1,\ldots,A^j\right) \mid \pi^Y \right).
    \end{align*}
\end{restatable}

\begin{proof}
We have
\begin{align*}
    \util_M(\pi^N_A \mid \pi^Y) - q_{\left(\pi^Y \right)}\util_M\left(\pi^A \mid \pi^Y\right) \leq& 
    \util_M(\pi^N_A \mid \pi^Y) - q_{\left(\pi^B \right)}\util_M\left(\pi^A \mid \pi^Y\right)  \\
    =&\util_M\left[\left(B^\pre,A^1,B^1,\ldots,A^k,B^k,A^\suff\right)_A \mid \pi^Y \right] \\
    &- q_{\left(B^\pre,B^1,\ldots,B^k\right)} \util_M\left[\left(A^1,\ldots,A^k,A^\suff \mid \pi^Y\right)\right]    \\
    =&\left[\sum_{i=1}^k q_{\left(B^\pre, A^1,B^1,\ldots,A^{i-1},B^{i-1}\right)} 
        \util_M\left(A^i \mid \pi^Y\cup \{A^j\}_{j \in [i-1]} \right)\right] \\
    &+ 
        q_{\left(B^\pre, A^1,B^1,\ldots,A^{k},B^{k}\right)} \util_M\left(A^\suff \mid \pi^Y\cup A \right) \\
    &- q_{\left(B^\pre,B^1,\ldots,B^k\right)}
        \bigg{[}\sum_{i=1}^k q_{\left(A^1,\ldots,A^{i-1}\right)} 
        \util_M\left(A^i \mid \pi^Y\cup \{A^j\}_{j \in [i-1]} \right)\\
    &+
        q_{\left(A^1,\ldots,A^k\right)}\util_M\left(A^\suff \mid \pi^Y\cup A\right)\bigg{]}  
\end{align*}
where the first inequality holds since $q_{\left(\pi^B\right)} \geq q_{\left(\pi^Y\right)}$ (since $\pi^B \subseteq \pi^Y$), 
and since by \Cref{clm:u_M(piA mid piY) < 0} $\util_M\left(\pi^A \mid \pi^Y\right) < 0$. 
%(see inequality~\eqref{eq:u_M(piA mid piY) < 0}).  \tec{fix the reference to the inequality}

Now, for every $i$ we have: 
$$
q_{\left(B^\pre, A^1,B^1,\ldots,A^{i-1}, B^{i-1}\right)} = q_{\left(A^1,\ldots,A^{i-1},B^\pre,B^1,\ldots,B^{i-1}\right)} =
q_{\left(A^1,\ldots,A^{i-1}\right)}\cdot q_{\left(B^\pre, B^1,\ldots,B^{i-1}\right)}.
$$
Thus we can factor out $q_{\left(A^1,\ldots,A^{i-1}\right)}$ in the chain above, cancel out the ``$A^\suff$'' term and continue as follows:
\begin{align*}
\util_M(\pi^N_A \mid \pi^Y)& - q_{\left(\pi^Y \right)}\util_M\left(\pi^A \mid \pi^Y\right)   \\
\leq &\sum_{i=1}^k q_{\left(A^1,\ldots,A^{i-1}\right)}\cdot\left(q_{\left(B^\pre, B^1,\ldots,B^{i-1}\right)} - q_{\left(B^\pre,B^1,\ldots,B^k\right)}\right) 
        \util_M\left(A^i \mid \pi^Y\cup \{A^j\}_{j \in [i-1]} \right)    \\
=&\sum_{i=1}^k q_{\left(A^1,\ldots,A^{i-1}\right)}\cdot q_{\left(B^\pre, B^1,\ldots,B^{i-1}\right)}\left(1 - q_{\left(B^i,\ldots,B^k\right)}\right) 
        \util_M\left(A^i \mid \pi^Y\cup \{A^j\}_{j \in [i-1]} \right)    \\ 
=&\sum_{i=1}^k q_{\left(A^1,\ldots,A^{i-1}\right)}\cdot q_{\left(B^\pre, B^1,\ldots,B^{i-1}\right)}p_{\left(B^i,\ldots,B^k\right)} 
        \util_M\left(A^i \mid \pi^Y \cup \{A^j\}_{j \in [i-1]} \right)    \\ 
=&\sum_{i=1}^k q_{\left(A^1,\ldots,A^{i-1}\right)}\cdot q_{\left(B^\pre, B^1,\ldots,B^{i-1}\right)}
\left(\sum_{j=i}^k q_{\left(B^i,\ldots,B^{j-1}\right)}p_{\left(B^j\right)}\right) 
        \util_M\left(A^i \mid \pi^Y \cup \{A^j\}_{j \in [i-1]} \right)    \\
=&\sum_{j=1}^k \sum_{i=1}^j q_{\left(A^1,\ldots,A^{i-1}\right)}\cdot q_{\left(B^\pre, B^1,\ldots,B^{i-1}\right)}
 q_{\left(B^i,\ldots,B^{j-1}\right)}p_{\left(B^j\right)}
        \util_M\left(A^i \mid \pi^Y \cup \{A^j\}_{j \in [i-1]} \right)    \\  
=&\sum_{j=1}^k  q_{\left(B^\pre, B^1,\ldots,B^{j-1}\right)} p_{\left(B^j\right)} 
        \sum_{i=1}^j q_{\left(A^1,\ldots,A^{i-1}\right)} 
        \util_M\left(A^i \mid \pi^Y \cup \{A^j\}_{j \in [i-1]} \right)    \\
=&\sum_{j=1}^k  q_{\left(B^\pre, B^1,\ldots,B^{j-1}\right)} p_{\left(B^j\right)}  
        \util_M\left(\left(A^1,\ldots,A^j\right) \mid \pi^Y  \right)   \qedhere     
\end{align*}
\end{proof}

We plug the inequality in \Cref{clm:low_branch_inequality_manipulation} and inequality (\ref{eq:lhss}) into inequality (\ref{eq:low_branch_inequality}), to get
\begin{align}
    0 \leq\sum_{j=1}^k  q_{\left(B^\pre, B^1,\ldots,B^{j-1}\right)} p_{\left(B^j\right)}  
        \util_M
        %\util_Y
        \left(\left(A^1,\ldots,A^j\right) \mid \pi^Y \right) \label{eq:final_inequality}
\end{align}

Note that all the factors $q_{\left(B^\pre, B^1,\ldots,B^{j-1}\right)} p_{\left(B^j\right)}$
are strictly positive
since $\pi^Y$ does not have a deterministic box and $B$ is a subset of $\pi^Y$.
This implies that at least one of the expressions
$
\util_M
%\util_Y
\left(\left(A^1,\ldots,A^j\right) \mid \pi^Y \right)$, for $j \in [k]$, is non-negative.
Choose some $j$ that satisfies this.
To conclude the proof we would like to say that we can concatenate $\left(A^1,\ldots,A^j\right)$ to $\pi^Y$ without decreasing the utility, thus obtaining the desired contradiction,
analogously to what we did in 
the proof of \Cref{clm:u_M(piA mid piY) < 0}. 
%case (i).
The (small) problem is that there might be boxes $i \in \left(A^1,\ldots,A^j\right)$ for which $v_i < v_\A$.

To get around this, we note that
it cannot be the case that all boxes $i\in \left(A^1,\ldots,A^j\right)$ satisfy $v_i < v_\A$,
since the contribution of these boxes to
$\util_M \left(\left(A^1,\ldots,A^j\right) \mid \pi^Y \right)$ is strictly negative.
Consider then the impulsive strategy with dummies 
$\pi^\mathsf{rand} = \left(A^1,\ldots,A^j\right)_{\{i \mid v_i \geq v_\A\} \cap \left(A^1,\ldots,A^j\right)}$ 
which is obtained from $\left(A^1,\ldots,A^j\right)$ by replacing the inspection of every box $i$ for which $v_i < v_\A$ with a decision to halt with probability $p_i$ and otherwise continue to the next box.
Then we have $\util_M\left(\pi^\mathsf{rand} \mid \pi^Y\right) \geq 0$.
% $$
% \util_M\left(\left(A^1,\ldots,A^j\right)_{\{i \mid \bbe{v_i \geq v_\A}\}} \mid \pi^Y \right) 
% \geq \util_M \left(\left(A^1,\ldots,A^j\right) \mid \pi^Y \right) \geq 0.
% $$
Furthermore, by the observation above
this strategy has non-empty deterministic strategies in its support 
(recall that an impulsive strategy with dummies is a distribution over deterministic impulsive strategies).
Thus, there exists one such strategy, denoted $\pi^{\mathsf{diff}}$, for which $\util_M\left(\pi^{\mathsf{diff}} \mid \pi^Y \right) \geq 0$, and which satisfies $v_i \geq v_\A$ for every $i \in \pi^{\mathsf{diff}}$.
This in turn implies $\util_Y\left(\pi^{\mathsf{diff}} \mid \pi^Y \right) \geq 0$, by \Cref{obs:bottom_equals_top_minus_p_pi_v_pi}.
We now concatenate $\pi^{\mathsf{diff}}$ to $\pi^Y$ without decreasing the utility,
analogously to what we did in 
the proof of \Cref{clm:u_M(piA mid piY) < 0},
%case (i),
and get a contradiction to the definition of $\piStar$.
This concludes the proof of Lemma~\ref{lem:one_fork_wlog}.

\section{Reduction to Bernoulli Instances}
\label{sec:reduction}

In this section, we show how Theorem~\ref{thm:MTT_Submodular} implies that 
for any instance with arbitrary distributions and a submodular cost function
there is an optimal strategy with a fixed-order. 
We do so by transforming an instance with arbitrary distributions, to an instance of the problem with Bernoulli distributions.
 
We first discretize the support of the distributions using a discretization parameter $\epsilon$ and 
by capping the values by a sufficiently large number (that depends on the distributions and on $\epsilon$).  This leads to a modified instance with finite support. Then, we replace each box with a finite number of boxes 
with weighted Bernoulli distributions. 

Both transformations maintain several key properties of the instance. The goal of these transformations is to modify the instance to have only a finite number of weighted Bernoulli boxes, for which we can apply Theorem~\ref{thm:MTT_Submodular}.

\paragraph{Transformation 1:}
Transformation  $\transformeps{\epsilon}$, defined by a parameter $\epsilon>0$,  proceeds as follows: given an instance  $\inst=(D_1,\ldots,D_n,c)$, 
%with distributions with finite expectations
%\mfc{remember to explain why we assume finite expectation}\tec{we should either say it in the model or in the intro},
%with support  $\supp$,  
%$\transformeps{\epsilon}$ first defines a constant 
let $\kappa_\epsilon:= \min \{\kappa\geq 0 \mid  \sum_{i=1}^n \E{(V_i-\kappa)^+} \leq \epsilon \} $. 
 % \ffc{$\kappa_\epsilon$ may depends on $n$, so it is not properly a constant }\tec{I interpret it as a constant of that may depend on the instance. it depends also on the support}. 
Such a constant $\kappa_\epsilon$ is well defined for every $\epsilon>0$ %for a finite set of distributions with finite expectations 
since $\sum_{i=1}^n \E{(V_i-\kappa)^+  }$ is a monotone continuous decreasing function in $\kappa$, and the limit as $\kappa$ approaches infinity is $0$ (we refer the interested reader to Lemma~\ref{lem:kappa} in the appendix for a formal claim.)
%\bbc{Explain why the function is continuous, Lemma~\ref{lem:kappa} does not mention this.} \tec{I think it is clear enough from the formula. You have finite sum of continuous functions}

Using $\kappa_\epsilon$, for every $i$,
%$\transformeps{\epsilon}$  defines 
$D_i^\epsilon$ 
is defined
to be the distribution of the random variable  $\bar{V_i}= \epsilon\cdot \lfloor\frac{\min(V_i,\kappa_\epsilon)}{\epsilon}\rfloor$.
We remark that since $\kappa_\epsilon$ is finite, then the support of the new set of distributions is finite.
Finally, the output of $\transformeps{\epsilon}$ is $\transformeps{\epsilon}(D_1,\ldots,D_n,c)= (D_1^\epsilon,\ldots,D_n^\epsilon,c)$.

\begin{proposition}
\label{prop:trans1}
For every instance $\inst$ and every $\epsilon>0$, let $\inst^\epsilon = \transformeps{\epsilon}(\inst)$. Then the following properties hold:
\begin{enumerate}
    \item    For every strategy $\pi$ on instance $\inst$ there exists a strategy $\pi'$ on instance $\inst^\epsilon$ such that  $\utili{\inst}{\pi} \leq \utili{\inst^\epsilon}{\pi'} +2\epsilon$.
    \item     For every strategy $\pi'$ on instance $\inst^\epsilon$ there exists a strategy $\pi$ on instance $\inst$ such that  $\utili{\inst}{\pi} \geq \utili{\inst^\epsilon}{\pi'}$. 
      %Plus, 
      Furthermore, if $\pi'$ is a fixed-order strategy, then there exists such a fixed-order strategy $\pi$. 
\end{enumerate}
\end{proposition}
\begin{proof}

    The two instances $\cI$ and $\cI^{\epsilon}$ are on the same boxes, the only difference is that the random variables $V_i$ of $\cI$ are discretized as $\overline{V}_i$ in $\cI^{\epsilon}$. To prove the first part of the proposition we show how to construct a strategy $\pi'$ using $\pi$. When $\pi$ prescribes to open box $i$, strategy $\pi'$ does it and observes a realization $\bar{V_i}=\bar{v_i}$. Then one can draw $V_i$, according to the distribution $D_i$ conditioned on the event that $\epsilon \cdot \lfloor \frac{\min(V_i,\kappa_\epsilon)}{\epsilon}\rfloor= \bar{v_i}$, and keep playing as if $\pi$ saw the realization of $V_i$. 
    
    It is clear that if $\pi'$ would have received the value of $V_i$ (instead of $\bar{V_i}$), then $\pi'$ would have had the same performance as $\pi$. As this is not always the case we bound the difference of the two utilities partitioning the analysis in two cases. If the chosen $V_i$ is at most $ \kappa_\epsilon$ then $V_i-\bar{V_i} \leq \epsilon$.
Otherwise, by the choice of $\kappa_\epsilon$, if we don't count $V_i$ in this event, we lose at most an additional $\epsilon$.
Formally, if we denote by $i^*$ the box with the maximal simulated 
value observed by $\pi'$, 
then:
\begin{align*}
\utili{\inst^\epsilon}{\pi'}  =&  \utili{\inst}{\pi} - \E{ V_\istar - \epsilon \cdot \left\lfloor \frac{\min(V_\istar,\kappa_\epsilon)}{\epsilon}\right\rfloor} \\ 
=& \utili{\inst}{\pi} - \E{ \left(V_\istar - \epsilon \cdot \left\lfloor \frac{V_\istar}{\epsilon}\right\rfloor \right) \cdot \indicator{  V_\istar \leq \kappa_\epsilon} } \\ 
&- \E{ \left(V_\istar - \epsilon \cdot \left\lfloor \frac{\kappa_\epsilon}{\epsilon}\right\rfloor\right) \cdot \indicator{ V_\istar > \kappa_\epsilon}} \\
\geq&
\utili{\inst}{\pi} - \E{ \epsilon \cdot \indicator{  V_\istar \leq \kappa_\epsilon} }\\
&- \E{ \left(V_\istar - \kappa_\epsilon\right) \cdot \indicator{ V_\istar > \kappa_\epsilon}} -\E{ \left(\kappa_\epsilon - \epsilon \cdot \left\lfloor \frac{\kappa_\epsilon}{\epsilon}\right\rfloor\right) \cdot \indicator{ V_\istar > \kappa_\epsilon}} \\
\geq & 
\utili{\inst}{\pi} - \E{ \epsilon \cdot \indicator{  V_\istar \leq \kappa_\epsilon} }\\ & - \left[\sum_i \E{ \left(V_i - \kappa_\epsilon\right) \cdot \indicator{ V_i > \kappa_\epsilon}}\right] -\E{ \epsilon \cdot \indicator{ V_\istar > \kappa_\epsilon}}\\ 
\geq& \utili{\inst}{\pi} - 2\epsilon,
\end{align*}
where the first two inequalities follows from the fact that, for every $x$, it holds that $x-\epsilon\left\lfloor\frac{x}{\epsilon}\right\rfloor \leq \epsilon$, and the last inequality is by definition of $\kappa_\epsilon$.

The second part of the proposition follows by the simple observation that one can run the following strategy $\pi$. First, calculate $\kappa_\epsilon$. Upon the arrival of $V_i=v_i$, calculate $ \bar{v_i}=\epsilon\cdot \left\lfloor\frac{\min(v_i,\kappa_\epsilon)}{\epsilon}\right\rfloor$.
Play according to $\pi'$ as if observed the value $\bar{v_i}$.
As the actual value of $v_i$ is always at least the value $\bar{v_i}$, it holds that $\utili{\inst}{\pi} \geq \utili{\inst^\epsilon}{\pi'}$. The ``furthermore'' part follows immediately by the structure of strategy $\pi$. 
\end{proof}

\paragraph{Transformation 2:}
Transformation $\transformber$ receives an instance $\inst=(D_1,\ldots,D_n,c)$ with distributions with finite supports, and returns a
%binary
Bernoulli
instance by the following process: 
%Let $m$ be the size of the union of the supports of $D_1,\ldots,D_n$. 
We can assume without loss of generality that $0$ is in the union of the supports $\supp$, then, we can rename the elements of the union of the supports in an increasing order $\supp=\{v_1,\ldots,v_m\}$, where $0=v_1<v_2<\ldots<v_m$.
%\bbc{Why is this w.l.o.g.?  Do you simply add 0 to the support, with probability 0 if it's not really in the support?} \tec{you can add things to the support and have $0$ probability of seeing them and it does not change anything}
For every $i\in[n]$ and $j\in[m]$, let $D_{i,j}$ be the weighted Bernoulli distribution that returns the value $v_j$ with probability $\frac{\P{V_i =v_j}}{\P{V_i\leq v_j}}$, and $0$ otherwise (where $\frac{0}{0}$ is interpreted as $0$). Let $c':2^{[n]\times [m]}  \rightarrow \mathbb{R}_{\geq 0}$ be the cost function where for every $S\subseteq [n] \times  [m]$, $$c'(S):= c(\{i \mid \exists j\in [m]\mbox{ such that } (i,j)\in S \}).$$ 
Then $\transformber(\inst)=(D_{1,1},\ldots,D_{n,m},c')$.
One can easily verify that $\transformber$ maintains monotonicity and normalization of the cost function. The following claim shows that it also maintains submodularity of the cost function.
%We first note that the transformation $\transformber$ maintains several properties of the cost function.
\begin{proposition}
\label{prop:maintain-sm}
If $c$ is submodular, then $c'$ obtained by transformation $\transformber$ is also submodular. 
\end{proposition}
\begin{proof}
We prove that for any pair of sets $S,T$ such that $S \subseteq T \subseteq [n] \times [m]$, and any  pair $(i,j) \notin T$, it holds that $c'(\{(i,j)\}\cup S) - c'(S) \geq c'(\{(i,j)\}\cup T) - c'(T) $.
Let $A_S := \{i' \mid \exists j'\in [m]\mbox{ such that } (i',j')\in S \}$, $A_T := \{i' \mid \exists j'\in [m]\mbox{ such that } (i',j')\in T \}$, and note that $A_S \subseteq A_T$.
If $i\in A_S$, then it holds that $c'(\{(i,j)\}\cup S) = c'(S) $ and $c'(\{(i,j)\}\cup T) = c'(T) $, thus $c'(\{(i,j)\}\cup S) - c'(S) =0=  c'(\{(i,j)\}\cup T) - c'(T) $.
 Else, it holds that $c'(\{(i,j)\}\cup S) - c'(S) =c(A_S\cup \{i\}) -c(A_s) \geq c(A_T\cup \{i\})-c(A_T)=  c'(\{(i,j)\}\cup T) - c'(T) $, where the inequality follows by submodularity of $c$.
\end{proof}

In Appendix~\ref{app:reduction}, we show that $\transformber$ maintains also MRF (Claim~\ref{cl:mrf}), GS (Claim~\ref{cl:gs}), coverage (Claim~\ref{cl:coverage}), XOS (Claim~\ref{cl:xos}) and subadditivity (Claim~\ref{cl:subadditive}) of the cost function, but not budget additive (Claim~\ref{cl:budget}). 
%\ffc{I suggest to merge Appendix A (definition of MRF and GS) and Appendix C}
%\bbc{Consider changing to propositions as above}

We next show that the new instance  $\inst'=\transformber(\inst)$ is equivalent to $\inst$ in the following sense:

\begin{restatable}{proposition}{transtwo}
\label{prop:trans2}
For every instance $\inst$, let $\inst'= \transformber(\inst)$. Then:
\begin{enumerate}
    \item    For every strategy $\pi$ on instance $\inst$, there exists a strategy $\pi'$ on instance $\inst'$ such that  $\utili{\inst}{\pi} \leq \utili{\inst'}{\pi'}$.
    \item     For every strategy $\pi'$ on instance $\inst'$, there exists a strategy $\pi$ on instance $\inst$ such that  $\utili{\inst}{\pi} \geq \utili{\inst'}{\pi'}$.
    Furthermore, if $\pi'$ is impulsive, then there exists such a fixed-order strategy $\pi$. 
\end{enumerate}
\end{restatable}
%\transtwo*
\begin{proof}
For the first part of the proposition, consider the following $\pi'$ that simulates $\pi$: every time $\pi$ opens a box $i$, the strategy $\pi'$ opens the corresponding set of boxes $\{(i,j)\}_j$ that are created by the transformation (in arbitrary order). Let $\bar{u}_{i,j}$ be the realized value from box $(i,j)$, then $\pi'$ behaves as if $\pi$ observed the value $u_i=\max_{j} \bar{u}_{i,j}$. The distribution of $u_i$ is exactly  $D_i$ since 
% \ffc{Is there a reason to use eqnarray* and not align*?}\tec{no}
\begin{align*}
\P{u_i=v_j } &= \frac{\P{V_i =v_j}}{\P{V_i\leq v_j}} \cdot \prod_{k>j} \left( 1-\frac{\P{V_i =v_k}}{\P{V_i\leq v_k}} \right) =\P{V_i=v_j} \cdot \frac{\prod_{k>j} \P{V_i<v_k}}{\prod_{k\geq j} \P{V_i \leq v_k}} \\ &= \P{V_i=v_j} \cdot \frac{\prod_{k>j} \P{V_i<v_k}}{\P{V_i\leq v_m} \cdot \prod_{k> j} \P{V_i < v_k}} = \P{V_i=v_j}, 
\end{align*}
 % \ffc{Why does the last equality hold?} \tec{I added a step that helps seeing it. Is it better now?}\ffc{Yes, better. I see that it is a telescopic argument on the denumerator}
where the third equality is since $\P{V_i\leq v_{k}} = \P{V_i <v_{k+1}}$.
Thus, the strategy guarantees the same expected utility (as both the distributions of the costs and the values are the same in $\pi$ and $\pi'$).
%\bbc{Say something about the costs being the same?}\tec{I think it is not necessary. Fede - what do you think?}. 
For the second part of the proposition, given a strategy $\pi'$, consider the strategy $\pi$ that simulates $\pi'$ by the following process: Whenever $\pi'$ tries to open $D_{i,j}$, if $D_i$ was not already open, then open $D_i$ (otherwise don't open anything). 
If the value $v_k$ was observed from $D_i$ and $k=j$ then  $\pi$ behaves as if the value $v_j$ was observed from $D_{i,j}$. If $k<j$  then  $\pi$ behaves as if the value $0$ was observed from $D_{i,j}$. 
Otherwise ($k>j$) then $\pi$ draws a sample  $s_{i,j}$ from $ D_{i,j}$ and behaves as if this value was observed.
The probability overall that 
 $\pi$  simulates that $ D_{i,j}$ was non-zero is 
 \[
    \P{V_i=v_j}+\P{V_i > v_j} \cdot \frac{\P{V_i=v_j}}{\P{V_i \leq v_j}}  = \frac{\P{V_i=v_j}}{\P{V_i \leq v_j}}.
 \]
% \bbc{I would add an intermediate step:  $\P{V_i > v_j} = 1- \P{V_i \leq v_j}$.}
 The cost of $\pi$ is always the same as that of $\pi'$, but its value can only be larger 
 (since $\pi$ never pretends to see a larger value than what it actually observed).
 The ``furthermore" part follows by observing that box $i$ is opened when the first box of the form $(i,j)$ is supposed to be opened by $\pi'$. Thus, box $i$ is opened before box $i'$ if and only if the first copy of $i$  is opened before the first copy of $i'$ in  the instance $\inst'$.
\end{proof}
In Section \ref{sec:MTT} we showed that for weighted Bernoulli instances with submodular costs, there exists an optimal strategy that is impulsive. We next show that this implies our main theorem: 
\begin{theorem}
\label{thm:main}
For every instance $\inst=(D_1,\ldots,D_n,c)$ where $c$ is  submodular, there exists an optimal strategy that is a fixed order strategy with thresholds. 
\end{theorem}
\begin{proof}
  Let $
\piStar$ be an optimal strategy for $\inst$ and let $\pi$ be the optimal strategy for $\inst$ among the strategies with a fixed order.
Assume towards contradiction that $\utili{\inst}{\piStar} > \utili{\inst}{\pi}$.
Let $\epsilon= \frac{\utili{\inst}{\piStar} -\utili{\inst}{\pi}}{4}$, and let $\inst^\epsilon=\transformeps{\epsilon}(\inst)$. 
By Proposition~\ref{prop:trans1}, there exists $\pi_1$ such that $\utili{\inst}{\piStar} \leq \utili{\inst^\epsilon}{\pi_1} +2\epsilon$.
Let $\inst'=\transformber(\inst^\epsilon)$.
Then, by proposition~\ref{prop:trans2} there exists $\pi_2$ such that  $\utili{\inst^{\epsilon}}{\pi_1} \leq \utili{\inst'}{\pi_2}$.
By Theorem~\ref{thm:MTT_Submodular} there exists  an impulsive strategy $\pi_3$ such that $\utili{\inst'}{\pi_2} \leq \utili{\inst'}{\pi_3}$.
By Proposition~\ref{prop:trans2} there exists a fixed-order  $\pi_4$ such that  $\utili{\inst^{\epsilon}}{\pi_4} \geq \utili{\inst'}{\pi_3}$, and by Proposition~\ref{prop:trans1} bthere exists a fixed-order  $\pi_5$ such that  $\utili{\inst}{\pi_5} \geq \utili{\inst^{\epsilon}}{\pi_4}$.
All together we have:
$$
\utili{\inst}{\pi_5} \geq \utili{\inst^{\epsilon}}{\pi_4} \geq \utili{\inst'}{\pi_3}\geq \utili{\inst'}{\pi_2} \geq \utili{\inst^{\epsilon}}{\pi_1} \geq  \utili{\inst}{\piStar} -2\epsilon  > \utili{\inst}{\pi},
$$ which contradicts the assumption that $\pi$ is  the optimal fixed-order strategy.
\end{proof}

\section{Computational Results}
\label{sec:computational}

    In this section we show that the task of finding an optimal strategy for Pandora's problem with submodular costs does not admit a polynomial time algorithm.  In fact, we show a stronger result, namely that there exists no algorithm for the Pandora's {\em decision problem} that uses a polynomial number of cost queries. 
    % Recall \mfc{Why "recall"? did we define it anywhere?}\ffc{Changed}, that a
    An algorithm for the Pandora's decision problem takes as input an instance of the Pandora's problem with a combinatorial cost function and outputs whether there exists a strategy yielding strictly positive utility on that instance.\footnote{
    We remark that even a demand oracle to the cost function, in the sense of \citet{blumrosen07}, would not allow us to solve the decision problem with polynomially many queries. The reason is that our impossibility result already holds for matroid rank functions, a strict subclass of gross substitutes, for which a demand query can be simulated by polynomially many cost queries; see \Cref{app:combinatorial}
    for definitions of gross substitutes and matroid rank functions.}

    \subsection{Distinguishing Submodular Functions}

        To formalize our argument we use the notion of distinguishability of submodular functions, as it is introduced in \citet{SvitkinaF11}. 
        We say that an algorithm distinguishes between two cost functions $c_1$ and $c_2$ if it produces different outputs when given  oracle access to $c_1$ versus oracle access to $c_2$. Here, we construct a family of cost functions and a baseline cost function that are hard to distinguish using polynomially many cost queries, similarly to the construction of \citet{SvitkinaF11}. Let $X$ be a set of $n$ boxes, and let $\alpha = \left\lceil \ln n \cdot \frac{ \sqrt{n}}{5}\right\rceil$ and $\beta = \left\lceil\frac {\ln^2 n}5\right\rceil$. On this set of boxes we define a ``baseline'' cost function $c_0(S) = \min\{|S|, \alpha\}$. Then, for any subset $R\subseteq X$ of boxes with $|R| = \alpha$, we define the cost function $c_R$:
        \begin{equation}\label{eq:c_R}
            c_R(S) = \min\{|S|, \alpha, \beta + |S \cap R^C|\}.
        \end{equation}
        It is immediate to see that $c_0$ and $c_R$ are submodular and differ on sets $S$ such that $\beta + |S \cap R^C|$ is strictly smaller than $\min\{\alpha,|S|\}$. Consider now a random set $\mathcal R$ that is drawn uniformly at random from all the subsets of $X$ of cardinality $\alpha$. It is possible to show that no deterministic algorithm can distinguish $c_R$ (for a random set $R\sim \mathcal R$) from $c_0$, with high probability. % (with respect to the random choice of $R$).
        We formalize this result in the following theorem.
        \begin{theorem}
        \label{thm:distinguish}
            Let $\mathcal A$ be any deterministic algorithm
            %Consider any deterministic algorithm $\mathcal A$ 
            that has a cost oracle access to a submodular function $c_{R} \sim c_{\mathcal R}$ over a set $X$ of $n$ elements, which outputs a set $S\subseteq X$ using polynomially many cost queries. Then, for any sufficiently large $n$, %the following holds: 
            \[
                \P{\mathcal A \text{ distinguishes $c_{\mathcal R}$ from $c_0$}} \le n^{-1}.
            \]
        \end{theorem}
        \begin{proof}
        First, we show that for any deterministic set $S$, the event that $c_0(S) \neq c_R(S)$ is negligible, with respect to the random draw of $R \sim \mathcal R$. 
        We observe that the probability of this event is maximized for sets $S$ of cardinality $\alpha$, so it suffices to restrict attention to such sets. 
        To see this, 
        %Such probability is maximized for sets $S$ of cardinality $\alpha$, so we only study these cases.
        %To see why the probability of $c_0(S) \neq c_R(S)$ is maximized for sets of cardinality $|S|=\alpha$, 
        we consider two cases. If $|S| \ge \alpha$, then $c_0(S) \neq c_R(S)$ if and only if $|S \cap R^C| < \alpha - \beta$, which is more likely to hold for small $|S|$. Conversely, if $|S| \le \alpha$, then $c_0(S) \neq c_R(S)$ when $|S|> |S \cap R^C| + \beta$, which is equivalent to $|S\cap R| > \beta$; it is easy to see that the latter condition is more likely when $|S|$ is large.
        Moreover, for $S$ of cardinality $\alpha$, the function $c_0$ and the realized function $c_R$ disagree if and only if $|S \cap R| > \beta$.

        To simplify the calculations, consider a set $R'$ that is obtained independently from $R$, sampling each element with probability $\alpha/n$. We have:
        \begin{align*}
            \P{|S \cap R'| > \beta} &= \sum_{k=0}^n \P{|R'|=k} \P{|S \cap R'| > \beta \mid |R'|=k} \\
            &\ge \P{|R'|=\alpha} \P{|S \cap R'| > \beta \mid |R'|=\alpha}\\
            &\ge \frac{1}{n^2}\P{|S \cap R'| > \beta \mid |R'|=\alpha},
        \end{align*}
        where the last inequality follows from the fact that $|R'|$ can attain $n+1 < n^2$ different values and $|R'| = \alpha$ is the most likely of them.
        We can use this argument on $R'$ to upper bound the probability of the event that ${|S \cap R| > \beta}$ in a simpler way:
        \begin{align*}
            \P{|S \cap R| > \beta} &= \P{|S \cap R'| > \beta| |R'|=\alpha} \le n^2 \cdot \P{|S \cap R'| > \beta}.
        \end{align*}
        We can now 
        focus on the term in the right-hand side, which is more amenable to Chernoff bound.
        %So, we move our attention to the last term, where we can apply more easily Chernoff bound. 
        The expected cardinality of $S \cap R'$ (with respect to the random choice of $R'$) is $\mu = \frac{\alpha \cdot |S|}{n} = \frac{\alpha^2}{n}$, while $\beta = 5 \mu,$ 
            % \mfc{how did you get this?} \ffc{Definition of $\alpha$ and $\beta$ given before the statement of the Theorem}
        we have:
        \[
            \P{|S \cap R'| > \beta} < \left(\frac{e^{\delta}}{(1+\delta)^{\delta}}\right)^\mu = \left(\frac{e^{4}}{25}\right)^{\alpha^2/n} \le 0.851^{\alpha^2/n}
        \]
        It follows that for all $S$ such that $|S| = \alpha$, it holds that
        % \begin{equation}
        % \label{eq:prob_bound}
        \[
            \P{|S \cap R| > \beta} \le n^2 \cdot 0.851^{\alpha^2/n}.
        \]
        % \end{equation}
        Consider now any deterministic algorithm $\mathcal A$ that performs at most polynomially many cost queries, and the computation path it follows when receiving as input the baseline cost function $c_0$. Both $\mathcal A$ and $c_0$ are deterministic, thus this is a single computation path. 
        Along this path, the algorithm performs at most a polynomial number of cost queries (say at most $n^a$, for some constant $a$), and each cost query distinguishes $c_0$ from $c_R$ with probability at most $n^2 \cdot 0.851^{\alpha^2/n}$ (as we have shown in the first part of the proof), thus by the union bound we get: 
        % %We know that the algorithm performs at most a polynomial number of cost queries, let's say $n^a$ for some constant $a$. 
        % \ffe{The probability (with respect to the randomness of $\mathcal R$)}
        % The probability (with respect to the randomness of $\mathcal R$) that the algorithm 
        % returns a different output
        % %outputs something different 
        % when given access to the realized $c_R$ instead of $c_0$ is thus upper bounded by $n^a$ times the probability that any fixed set $S$ receives is distinguishable \mfc{last sentence unclear}. 
        % %All in all, 
        % Hence, we get
        \[
            \P{\text{$\mathcal A$ distinguished $c_0$ from $c_{\mathcal R}$}} \le n^{a+2}\cdot 0.851^{\alpha^2}\le n^{a + 2 + \ln(0.851)\ln n}<\frac 1n,
        \]
        where the last inequality holds for all $n>\frac{a+3}{-\ln(0.851)}$. 
        The latter condition specifies what we mean by ``sufficiently large $n$'' in the statement of the theorem: for any fixed algorithm $\mathcal A$ that performs $O(n^a)$ queries, taking $n>\frac{a+3}{-\ln(0.851)}$ gives the desired claim for $\mathcal A.$
    \end{proof}

% \mfc{Any particular reason to have it in an observation environment? Unless it's referred to from elsewhere, let's just write this in text ("Note that, ...".}

    Note that, by the proof of Theorem~\ref{thm:distinguish},
        $n^{-1}$ as a bound on the probability 
        that $\mathcal A$ distinguishes $c_{\mathcal R}$ from $c_0$ can be replaced by $n^{-b}$ for any constant $b$. The corresponding ``sufficiently large $n$'' condition would then be $n>\frac{a+2+b}{-\ln(0.851)}$.
    
    %\begin{observation}
        % \mfe{By the proof of Theorem~\ref{thm:distinguish},
        % $n^{-1}$ as a bound on the probability 
        % that $\mathcal A$ distinguishes $c_{\mathcal R}$ from $c_0$ can be replaced by $n^{-b}$ for any constant $b$. The corresponding ``sufficiently large $n$'' condition would then be $n>\frac{a+2+b}{-\ln(0.851)}$.}
        %it is clear that the choice of $n^{-1}$ as a bound on the probability 
        %\mfe{that $\mathcal A$ distinguishes $c_{\mathcal R}$ from $c_0$}
        %of $\mathcal A$ distinguishing $c_{\mathcal R}$ to $c_0$ 
        %is arbitrary; \mfe{indeed, for any constant $b$, $n^{-b}$ would work, and the corresponding ``sufficiently large $n$'' would be $n>\frac{a+2+b}{-\ln(0.851)}$.}
        %, for constant $b$, would have worked, with the only effect of changing the notion of ``sufficiently large $n$'' to $n>\frac{a+2+b}{-\ln(0.851)}$.
     %\end{observation} 

    \subsection{A Family of Difficult Instances}

        As we show next, the family of submodular cost functions introduced above induces a family of instances of Pandora's problem such that $(i)$ the baseline instance admits no strategy that gives positive utility, and
        %no strategy on the baseline instance achieves positive utility, and 
        $(ii)$ every other instance in the family admits a strategy obtaining positive utility. 
        
        Formally, fix any large enough $n$ and consider the following class of instances of Pandora's Problem with submodular cost functions: there is a set $X$ of $n$ boxes
        %are $n$ boxes $X$ 
        with i.i.d. values distributed according to the following weighted Bernoulli distribution: 
        the value of every box in $X$ is $M = 5 \beta>0$ with probability $p= \frac 1\alpha$, and $0$ otherwise.
        %that follow the weighted Bernoulli distribution in $\{0, M\}$, for $M = 5 \beta>0$ and probability of success $p= \frac 1\alpha$. 
        For each $R\subseteq X$, with $|R|=\alpha$, we define the instance $\mathcal{I}_R$ with the above random values and the cost function $c_R$ that is given in Equation~\eqref{eq:c_R}. Moreover,  we construct the baseline instance $\mathcal I_0$ using the same random variables, but with cost function $c_0$. There is a crucial difference between $\cI_0$ and $\cI_R$: With $c_R$ it is possible to find a subset of $\alpha$ boxes such that only the first $\beta$ of them have non-zero marginal cost, while this is impossible under $c_0$. With our choice of $M$ and $p$ it is possible to leverage this property and show the following Lemma.
        
        \begin{lemma}
        \label{lem:family}
            For any sufficiently large $n$, no strategy extracts positive utility from $\mathcal I_0$, while for any $R$ there exists a strategy that extracts positive utility from $\mathcal I_R$.
        \end{lemma}
        \begin{proof}
        We first establish the second part of the lemma.
            Consider any $\mathcal I_R$, and the (optimal) strategy $\pi^R$ for $\mathcal I_R$, which knows the specific set $R$. The strategy  $\pi^R$ opens the boxes in $R$ one after the other (in any order) and halts when the value $M$ is realized for the first time, and otherwise when all boxes in $R$ are exhausted.
            
            The expected reward of $\pi^R$ is computed as follows.
            %this strategy is easy to compute: 
            The value $M$ is achieved if at least one of the $\alpha$ Bernoulli boxes in $R$ is realized, 
            %in the boxes in $R$ is realized, 
            yielding an expected value of $M(1-(1-p)^{\alpha})$. On the other hand, their total cost is
            %cost paid is 
            at most $\beta$. 
            We get:
            \begin{align*}
                \utili{\mathcal I_R}{\pi^R}\ge
                M\left(1-\left(1-\frac{1}{\alpha}\right)^{\alpha}\right) - \beta \ge 5\beta\left(1-\frac 1e\right) - \beta > 0.
            \end{align*}

            We next establish the 
            %We move now our attention to the 
            first part of the lemma. 
            First, there exists a deterministic strategy that is optimal for $\mathcal I_0$ (see Section \ref{sec:preliminaries}), thus we restrict attention to deterministic strategies. 
            Second, since all boxes are symmetric, there exists an optimal strategy that is impulsive (see Section~\ref{sec:setup}); i.e., it commits to a subset $S$ of the boxes, and opens them sequentially in an arbitrary fixed order, until $M$ is realized (or until all boxes in $S$ have been opened).
            Note that $c_0(S)$ depends only on the cardinality of $S$, so all the orderings are equivalent. 
Let $\pi^S$ denote this strategy.

To conclude the proof, we show that for every set $S$ we have $\utili{\mathcal I_0}{\pi^S} \le 0$.       
            %As above, we assume without loss of generality that the strategy halts in the first time where $M$ is realized.
            %immediately after the value $M$ is realized for the first time.
            We distinguish between four cases, depending on the cardinality of $S$. 
            
            {\bf Case 1:} $|S| \ge \alpha$. Such a policy opens all the boxes in $S$ (with expected reward $M(1-(1-p)^{|S|})$) but pays only for the first $\alpha$ of them. Using similar reasoning as above %Reasoning similarly to above 
            we get: 
            \begin{align*}
                \utili{\mathcal I_0}{\pi^S} &= M(1-(1-p)^{|S|}) - p \cdot\sum_{i=1}^{\alpha} i(1-p)^{i-1}  - \alpha(1-p)^{\alpha}\\
                &\le  M- \alpha\left(1-\frac 1\alpha \right)^{\alpha} \le M - \frac \alpha 4 = 5 \beta - \frac \alpha 4 < 0,
            \end{align*}
            %\bbc{Should be $\alpha(1-p)^{\alpha - 1}$ in the first line.}
            %\ffc{we pay $\alpha$ in two cases: if we open $\alpha-1$ boxes without $M$ but then the $\alpha^{th}$ box contains $M$ (captured by the last term in the sum) or if all the $\alpha$ boxes contain $0$ (last term of the first line)}
            %\bbc{Right.  The sum of the two terms is $\alpha(1-p)^{\alpha - 1}$}
            where the last inequality follows from the fact that $\alpha > 20 \beta$
            %we used the fact that 
            %$\alpha > 20 \beta$ 
            (recall that $\alpha \in \Theta(\ln n\sqrt{n})$, while $\beta \in \Theta(\ln^2 n)$).
            
            {\bf Case 2:} $|S| \in \{21 \beta, \dots, \alpha -1 \}$. Note that for $|S| < \alpha$, the cost function $c_0$ is simply additive, thus 
            \begin{align*}
                \utili{\mathcal I_0}{\pi^S} &= M(1-(1-p)^{|S|}) - p \cdot\sum_{i=1}^{|S|} i\cdot (1-p)^{i-1}  - |S|(1-p)^{|S|}\\
                &\le M - |S|\left(1-\frac{1}{|S|}\right)^{|S|} \le M - \frac{|S|}4 \le - \frac \beta 4< 0.
            \end{align*}
            % \bbc{$|S|(1-p)^{|S| - 1}$}\ffc{see above}
            
            {\bf Case 3:} $0 < |S| < 21 \beta$. We have:
            \begin{align*}
                \utili{\mathcal I_0}{\pi^S} &= M(1-(1-p)^{|S|}) - p \cdot \sum_{i=1}^{|S|} i\cdot (1-p)^{i-1}  - |S|(1-p)^{|S|}    \\
                &\le M - (M + |S|)(1-p)^{|S|}\\
                &\le M - (M + |S|)(1-p|S|)
                \\
                &= |S| \left(pM + p|S| - 1\right)\\
                &= p|S| \left(M + |S| - \frac 1p\right)\\
                &\le \frac{|S|}{\alpha}\left(26 \beta - \alpha\right) < 0,
            \end{align*}
            where the third line uses the Bernoulli inequality, and the last two inequalities use the definitions of $\alpha$ and $\beta$, the fact that $n$ is sufficiently large, and the condition of case 3 (i.e., $0<|S|<21 \beta$). 
            
            {\bf Case 4:} $|S| = 0$. This case corresponds to the strategy that does not do anything, whose utility is clearly $0$.  
        \end{proof}
        
    \subsection{The Computational Impossibility Result}

        We are ready for the main theorem of the section: since it is not possible to distinguish in polynomial time between $c_R$ and the baseline $c_0$, then it is not possible to assess, in polynomial time, whether an instance of Pandora's problem can yield positive utility (as $\mathcal I_R$) or not (as the baseline instance $\mathcal I_0$). This immediately implies that no computationally efficient approximation for Pandora's problem with Submodular cost exists. 

        To formalize our result we introduce the concept of positivity oracle: a (possibly randomized) algorithm $\mathcal O$ is a positivity oracle for Pandora's problem with Submodular cost if it takes in input an instance $\mathcal I$ of the problem (i.e. knowledge of the distributions of the random rewards and cost oracle access to the cost function) and outputs an answer to the question whether it exists or not a strategy $\pi$ such that $\utili{\mathcal I}{\pi}>0$. We say that $\mathcal O$ is correct on instance $\mathcal I$ with a certain probability $p$ if it outputs the correct answer to Pandora's decision problem on that instance with probability at least $p$, where the probability is with respect to the internal randomization of $\mathcal O$. In other words, a positivity oracle is an algorithm for Pandora's decision problem. We say that a positivity oracle $\mathcal O$ is efficient if there exists a constant $a$ (that depends on $\mathcal O$ but not the specific instance) such that $\mathcal O$ issues at most $n^a$ cost queries with probability $1$ on all instances. 

        \begin{theorem}
        \label{thm:oracle}
            Fix any efficient positivity oracle $\mathcal O$ and positive constant $\varepsilon>0$. Then there exists an instance $\mathcal I$ on $n$ boxes (for $n$ sufficiently large) such that $\mathcal O$ outputs the correct answer on $\mathcal I$ with probability at most $0.5 + \varepsilon$. 
        \end{theorem}
        \begin{proof}
            The possibly randomized positivity oracle $\mathcal O$ is just a distribution over deterministic ones, so for any $R \subseteq X$ of cardinality $\alpha$ and any deterministic positivity oracle $O$, we %define
            denote with $\mathcal E(O,R)$ the event that $O$ gives a different output when %receives in input $\mathcal I_R$ and $\mathcal I_0$.
            receiving $\mathcal I_R$ and $\mathcal I_0$ as input.

            Recall that $\mathcal R$ is a set of cardinality $\alpha$ drawn uniformly at random. 
            Denote with $\textbf{O}$ (respectively, $\textbf{O}_d$) the set of all the randomized (resp., deterministic) efficient positivity oracles. Yao's principle gives the following:
            \begin{align}
            \label{eq:Yao}
                \min_R \P{\mathcal E(\mathcal O,R)} \le \min_R \max_{\mathcal O^* \in \textbf{O}} \P{\mathcal E(\mathcal O^*,R)}\le \max_{ O \in \textbf{O}_d} \P{\mathcal E(O,\mathcal R)}.
            \end{align}
            Consider the rightmost term; each deterministic positivity oracle $O$ is an algorithm with cost oracle access to the underlying submodular cost function, which gives different outputs on $\mathcal I_R$ and $\mathcal I_0$ if it distinguishes $c_\mathcal R$ from $c_0$ (see definition of distinguishability), given that the rest of the input is identical.
            From \Cref{eq:Yao} we have then:
            \begin{equation}
            \label{eq:epsilon}
                \min_R \P{\mathcal E(\mathcal O,R)} \le \max_{ O \in \textbf{O}_d} \P{\mathcal E(O,\mathcal R)}\le  \frac 1n \le \varepsilon,
            \end{equation}
            where the second inequality follows from \Cref{thm:distinguish}, for any $n$ sufficiently large.

            What we have shown so far is that there exists a set $R$ such that $\mathcal O$ gives the same output on both $\mathcal I_0$ and $\mathcal I_R$ with probability at least $1-\varepsilon$ even though the correct answer to Pandora's decision problem on the two instances is different. Let now $\mathcal G_0$, respectively $\mathcal G_R$, be the event that $\mathcal O$ is correct on input $\mathcal I_0$, respectively $\mathcal I_R.$ If the probability of $\mathcal G_0$ is smaller than $0.5+\varepsilon$ then there is nothing else to prove, as we can choose $\mathcal I = \mathcal I_0$; otherwise, we have the following:
            \begin{align*}
                \P{\mathcal G_R} &= \P{\mathcal G_R \cap \mathcal E(\mathcal O,R)} + \P{\mathcal G_R \setminus \mathcal E(\mathcal O,R)} \le \P{\mathcal E(\mathcal O,R)} + \P{\mathcal G_0^C}\le  \varepsilon + (0.5 - \varepsilon) = 0.5.
            \end{align*}
            To see why the previous formula holds we study separately the two summands. The event $\mathcal G_R \cap \mathcal E(\mathcal O,R)$ is clearly contained in $\mathcal E(\mathcal O,R)$, and we know that its probability is smaller than $\varepsilon$ by \Cref{eq:epsilon}. The event $\mathcal G_R \setminus \mathcal E(\mathcal O,R)$, on the other hand, is disjoint from $\mathcal G_0$; in fact, we know that if $\mathcal O$ gives the same output for $\mathcal I_0$ and $\mathcal I_R$, at most one of the two instances receives the correct answer to its decision problem. Finally, we are under the assumption that $\P{\mathcal G_0} \ge 0.5 + \varepsilon$, thus its complementary has at most a probability $0.5-\varepsilon$ to realize. 
        \end{proof}

        The previous result directly implies that no approximation result can be achieved for Pandora's problem 
        with submodular cost functions using polynomially many cost queries: assume by contradiction that such an algorithm exists, then it would be easy to construct a positivity oracle that violates the previous theorem, e.g. by repeatedly simulating the algorithm and using concentration.

\section{Conclusion and Future Directions}
    In this paper we initiate the study of Pandora's problem with a combinatorial cost function, and study to what extent the simplicity of Weitzman's solution extends beyond additive cost functions. 
    We show that the structural simplicity carries over to submodular cost functions, but not to XOS cost functions. 
    Namely, Pandora's problem with submodular cost functions admits an optimal strategy that is fixed-order. 
    From a computational perspective, we prove that no polynomial-time approximation algorithm for the Pandora's problem with submodular cost function exists, even for the subclass of matroid rank functions.

    Our work suggests intriguing directions for future research. In particular, many of the variants of Pandora that have been studied under the classic model of additive cost functions can be studied under combinatorial cost functions. Obvious examples are settings beyond single choice and non-obligatory inspection. 
    In addition, some computational problems remain open. For example, it is not clear whether there exists a poly-time algorithm that finds an optimal (or an approximately optimal) strategy under budget-additive costs.

    % \ffc{Do we want to say that a characterization of the instances admitting a poly-time solution is missing? We have the negative result for simple matroid rank function, but positive result for another class of GS (precedence constraint).}

%\input{sections/80-acks.tex}

\newpage
% Bibliography
\bibliographystyle{plainnat}
\bibliography{bibliography}

\clearpage
% Appendix
\appendix

\section{Classes of Combinatorial Functions}
\label{app:combinatorial}

    In this appendix we recall the definitions and some properties of matroid rank functions, gross substitute functions and coverage functions. 

    \paragraph{Matroids.} A family of subsets $\mathcal M \subseteq 2^X$ of a base set $X$ is called a matroid if the following two properties hold: 
    \begin{itemize}
        \item{\em Downward closure}: if $A \in \mathcal M$ and $B \subseteq A$, then $B \in \mathcal M$
        \item{\em Augmentation property:} if $A,B \in \mathcal M$ and $|A|>|B|$, then there exists $x \in A\setminus B$ such that $B \cup \{x\} \in \mathcal M$
    \end{itemize}
    Subsets in $\mathcal M$ are called independent sets. 
    
    \paragraph{Matroid rank functions.} Given any matroid $\cM$, it is possible to define its associated rank function $r_{\cM}:2^X \to \mathbb{N}$ as the cardinality of the largest independent set: 
    \[
        r_{\cM}(A) = \max_{B \subseteq A}\{|B| \mid B \in \mathcal M\}.
    \]
    Recall the definitions of the two functions $c_0$ and $c_R$ in \Cref{sec:computational}. A base set $V$ of $n$ elements is given, as well as two integers $\alpha > \beta$ smaller than $n$ and a subset $R \subseteq V$, with $|R| = \alpha$, we have:
    \[
        c_0(S) = \min\{|S|,\alpha\}, \quad c_R(S) = \min\{|S|,\alpha, |S \cap R^C| + \beta\}.
    \]
    Cost function $c_0$ is the rank function of the matroid $\cM_0 = \{S \subseteq X \mid |S| \le \alpha\}$ (thus it is also submodular). For what concerns $c_R$, instead of explicitly exhibiting the relative matroid, we show that it respects two properties that ensure that there exists a matroid $\cM_R$ on $X$ of which $c_R$ is indeed the rank function. 

    \begin{theorem}[Theorem 39.8 of \citet{Schrijver03}]
    \label{thm:ranks}
        Let $V$ be a set and let $r : 2^X \to \mathbb{N}$. Then $r$ is the rank function of a matroid if and only if for all $T,U \subseteq V$:
        \begin{itemize}
            \item[$(i)$] $r(T) \le r(U) \le |U|$ if $T\subseteq U$
            \item[$(ii)$] $r(U \cap T) + r(U \cup T) \le r(U) + r(T)$.
        \end{itemize}
    \end{theorem}

    It is immediate to verify that $c_R$ respects property $(i)$ while condition $(ii)$ is equivalent to submodularity, which holds for $c_R$. 
    %\ffc{There is no need to spend more words on submodularity of $c_R$ I think.}
    % verify that $c_R$ is submodular, that is easy to  already know from the main text that this is indeed the case 
    % \bbc{This is not done in the main text.}.  \tec{we say in the main text that this is submodular. So if we need to explain it, it is in the main text, and I don't think we need to explain it}

    \paragraph{Gross-substitutes.} A function $f: 2^X \rightarrow \mathbb{R}_{\geq 0}$ is \emph{gross-substitutes} if for any two vectors 
    $p, p' \in \mathbb{R}_{\geq 0}^n$ (with $|X| = n$) such that
    $p' \geq p$ (component-wise) and any $S\subseteq X$ such that
    $S \in \arg\max_{T \subseteq X} f(T) - \sum_{i \in T} p_i$ there is a $S' \subseteq X$ such that $S' \in \arg\max_{T \subseteq X} f(T) - \sum_{i \in T} p'_i$ and 
    $ \{i \in S \mid p'_i = p_i\}\subseteq S'.$  It is known that the class of gross substitute functions strictly contains that of matroid rank functions \citep{BalkanskiL20} 
    
    An alternative definition of gross substitute functions is given by \citet{ReijnierseGM02} and is as follows:
    a valuation function $f$ has the gross substitutes property on set $X$ if and only if it is submodular, and for all sets $S \subseteq X$ and all distinct $i, j, k \in X \setminus S$, the following multi-set does not have a unique max: 
        \begin{equation}
            \label{eq:gs}
        \{f(\{i, j\}|S) + f(k|S),f(i|S) + f(j, k|S), f(j|S) + f(\{i, k\}|S)\}.
        \end{equation}

    \paragraph{Coverage.} A function $f: 2^X\rightarrow \mathbb{R}_{\geq 0}$ is a \emph{coverage} if there exist a set of elements $E$, a mapping of the elements to weights $w:E \rightarrow \mathbb{R}_{\geq 0} $, and a mapping $g:E\rightarrow 2^{X}$ such that $f(S) = \sum_{j\in E} w(j) \cdot \indicator{\exists i\in S \cap g(j) }$.    

\section{Adaptive Order is Necessary beyond Submodular Cost Functions}
\label{app:XOS}

\begin{claim} \label{cl:proof-example}
The only optimal strategy for the instance of Example~\ref{ex:adaptive-order} is to first open box $1$, and if the value of box $1$ is $10$, to open box $2$, and if the value of box $1$ is $0$, then open box $3$.
\end{claim}
\begin{proof}
As mentioned in Section~\ref{sec:preliminaries}, every randomized strategy is a distribution over deterministic strategies, and therefore it is sufficient to prove that  this  strategy is optimal among deterministic strategies.
We first claim that an optimal strategy will never open both boxes $2,3$. This holds since after opening one of them, then opening the other leads to a cost of $20$, but can only increase the value by at most $12$.
If the first box the strategy opens is $1$, then if its value is $10$, then the only box that has positive marginal utility is $2$, therefore the strategy should open it, and if its value is $0$, then since box $3$ has higher marginal utility (and we never open both $2,3$), then the strategy must open $3$.
If box $2$ is opened first, then if its value is $12$, all other boxes have non-positive marginal utility, and if its value is $0$, then  box $1$ is the only one with positive utility, so the strategy should open it.  This leads to a lower utility than the one that starts with opening box $1$. 
If box $3$ is opened first, then all other boxes have non-positive marginal utility, and this leads to a lower utility than the one that starts with opening box $1$. 
\end{proof}

\begin{claim}\label{cl:marginal-xos} 
    For every monotone and normalized function $f:2^X\rightarrow \mathbb{R}_{\geq 0}$, the function $g:2^{X\cup\{0\}} \rightarrow \mathbb{R}_{\geq 0}$ defined as follows is XOS:
    \[
        g(S) = n \cdot f(X)\cdot \indicator{S\neq \emptyset} + f(S\cap X) 
    \]
    Moreover, the marginal of $g$ with respect to the set $\{0\}$ is $f$ (i.e., for every $S\subseteq X$, it holds that $f(S) = g(S \cup \{0\}) - g(\{0\})$).
\end{claim}
\begin{proof}
To show that $g$ is XOS, we need to present additive functions such that the maximum over them is $g$.
For every $S\subseteq X$ we define an additive function $a^S$ over the elements in $X\cup\{0\}$ as follows:
If $S=\emptyset $, then $a^S(0)=n \cdot f(X)$, and $a^S(i)= 0$ for every $i\in [n]$.
If $S\neq\emptyset $, then 
%$a^S(i)=0 \cdot f(n)$
{$a^S(i) = 0$}
for {$i\in (X\setminus S)\cup\{0\}$}, and $a^S(i)= \frac{n \cdot f(X) +f(S)}{|S|}$ for  $i\in S$.
Observe that under this definition, for every $T \subseteq X \cup \{0\}$ we have
$$
a^{T\cap X}(T) = 
\begin{cases}
    \emptyset &T = \emptyset \\
    n \cdot f(X) + f(T \cap X) &\text{Otherwise}
\end{cases}
$$
That is, $a^{T\cap X}(T) = g(T)$.  Thus, in order to establish that $g=\max_{S \subseteq X} a^S$, it remains to show that
%\bbe{
%To show that $g=\max_{S \subseteq X} a^S$,
%we show that: 
%\begin{enumerate}
    %\item For every $T \subseteq X \cup \{0\}$, there exists $S \subseteq X$ such that $a^S(T) = g(T)$.
    %\item 
    for every $T \subseteq X \cup \{0\}$, and for every $S \subseteq X$, we have $a^S(T) \leq g(T)$.
%\end{enumerate}
%}
Let $T$ and $S$ be as described above.  If $T = \emptyset$, then the claim clearly holds (the inequality becomes $0 \leq 0$).  We thus assume that $T \neq \emptyset$.
Now, if $S = \emptyset$, then it is straightforward to see that
$$
a^{\emptyset}(T) =
\begin{cases}
    n \cdot f(X)  & 0 \in T \\
    0  &\text{Otherwise}
\end{cases}
\leq
n \cdot f(X)\cdot \indicator{T\neq \emptyset} + f(T\cap X) = g(T)
$$

If $S \neq \emptyset$, then $a^S(T) = a^S(T \cap S)$.  If we also have $S \subseteq T$, then
$a^S(T) =a^S(S)\leq g(T)$, by monotonicity of $f$.  If, on the other hand we have $ S \setminus T \neq \emptyset$, then
$$
a^S(T)  = \frac{\left|T \cap S\right|}{\left|S\right|}(n \cdot f(X) +f(S)) \leq n\cdot f(X) \leq g(T).
$$
This concludes the proof.
%To show that $g=\max_S a^S$,
%one can observe that for every $S$, it holds that 
%$a^S(S)= n \cdot f(X) +f(S)$.
%For every $S,T\neq \emptyset$, $a^S(T)= a^S(T\cap S)$, thus, if $S \subseteq T$, $a^S(T) =a^S(S)\leq g(T)$, and if $ S \setminus T \neq \emptyset$, then $a^S(T) \leq n\cdot f(X) \leq g(T)$.
%For $T=\emptyset$ $a^S(T)=0$. For $a^\emptyset$, it is only non-zero of sets that contain $\{0\}$.
%Overall we showed that it is always at most the value of $g$, and there exists one additive function where it is the value of $g$.
%The ``Moreover" part is derived immediately by the construction of $g$.
\end{proof}

\begin{claim} \label{cl:transform-to-xos}
Given an instance $\inst=(D_1,\ldots,D_n,c)$ where the optimal strategy must use adaptive order, there exists a distribution $D_0$, and an XOS cost function $c':2^
{[n]\cup\{0\}} \rightarrow \mathbb{R}_{\geq 0}$ where the optimal strategy for instance $\inst'=(D_0,D_1,\ldots,D_n,c')$ must use an adaptive order.
\end{claim}

\begin{proof}
    Consider an instance with an additional box (which we will refer to as box $0$).
The cost function $c'$ is the cost function described in Claim~\ref{cl:marginal-xos}.
The random variable $V_0$ for box~$0$ is defined as follows: 
We draw an independent sample $s_i$ from each $D_i$, 
and we also draw a Bernoulli random variable $b$ which equals 1 with probability $1/2$ and equals 0 otherwise.
Then $V_0$ is defined to be $2\cdot b\cdot (1+n\cdot c([n])+\max_i s_i)$.

Assume towards contradiction that there exists a deterministic optimal
strategy $\piStar_{\inst'}$ for instance $\inst'$ that uses a fixed-order, and let $\sigma:[n+1]\rightarrow[n]\cup\{0\}$ be that order.
Let {$\tau_1,\ldots,\tau_{n+1}$} be the optimal thresholds that $\piStar_{\inst'}$ uses (i.e., $\piStar_{\inst'}$ opens $\sigma(i)$ if up to round $i-1$ the maximum realized value so far is {strictly less than $\tau_{i}$}).

We first claim that {$\tau_1 > 0$} (\emph{i.e}, the strategy $\piStar_{\inst'}$ opens at least one box). If $\piStar_{\inst'}$ does not open any box,
then  $\piStar_{\inst'}$ can be strictly improved by opening box $0$ since the following holds true:
\[
    \E{V_0} = 2 \cdot \E{b} \cdot \E{1+n\cdot c([n])+\max_i s_i)} >  c(\{0\}).
\]

We next show that there is a fixed-order strategy with the same utility as $\piStar_{\inst'}$ that opens box $0$ first.
Let $i_0= \sigma^{-1}(0)$. If $i_0=1$, then $\piStar_{\inst'}$  opens box $0$ first, and we are done. 
Else, consider now the following fixed-order strategy $\pi$ with order $\sigma' = \left(0,\sigma(1),\ldots,\sigma(i_0-1),\sigma(i_0+1),\ldots,\sigma(n+1)\right)$:
The strategy opens box $0$, then it always opens box $\sigma(1)$. For $i=2,\ldots,i_0-1,i_0+1,\ldots,n+1$, it opens box $\sigma(i)$, if the maximum among the values of boxes $\sigma(1),\ldots,\sigma(i-1)$ is {strictly less than} $\tau_i$. Note that for $i<i_0$, the strategy ignores the value of box $0$ for the decision of whether to halt, {while for $i > i_0$ the strategy acts exactly as in $\piStar_{\inst'}$}.
For every realization of the values of the boxes, the costs of $\pi$ and $\piStar_{\inst'}$ are the same, and the value obtained by $\pi$ is at least as large as the value of $\piStar_{\inst'}$. Thus $\pi$ is an optimal strategy.

Under the case that the value that box $V_0 = 0$ (which happens with a probability of at least 0.5), the marginal instance that $\pi$ is facing is exactly $\inst$, and $\pi$ must be optimal also for this case (since it happens with a non-zero probability).
But since $\pi$ is a fixed-order strategy, its marginal strategy for this case is also a fixed-order strategy, and this contradicts that no fixed-order strategy is optimal for $\inst$.
This concludes the proof.
\end{proof}

    The previous claim, together with the \Cref{ex:adaptive-order} yields the following Theorem.
    
\begin{theorem}
\label{thm:XOSadaptive}
    There exists an instance of the Pandora’s problem with an XOS cost function that admits no optimal strategy with non-adaptive exploration order.
\end{theorem}

\section{Pandora's Problem with Order Constraints}
\label{app:pandora}
% \bbc{This section needs revision, and to be referenced from the body.}
% \ffc{We made a pass with Tomer and referenced to it in the Related work section and in the intro.}
    In this appendix, we show how some relevant features of Pandora's problem with order constraints \citep{BoodaghiansFLL20} are captured by suitable instances of Pandora's problem with combinatorial cost functions. 
    When the underlying precedence graph is a tree,  we show in \Cref{prop:gs,cl:submodular-equivalence} that it is possible to construct a gross-substitutes cost function such that the two problems (order constraint and combinatorial cost) are equivalent. This construction exhibits a non-trivial class of gross-substitutes cost functions for which the optimal strategy is fixed order and can be computed efficiently. Then, in \Cref{cl:subadditive-example},
    we borrow an instance from \cite{BoodaghiansFLL20} for which the optimal strategy is not fixed-order,
    and we translate it into an instance of Pandora's problem with a subadditive cost function that has the same property.
    This constitutes an alternative proof that adaptivity may be necessary under subadditive costs (other than the proof of Theorem \ref{thm:XOSadaptive}).
    %we use the construction of a general order constraint where the optimal strategy on it needs to be adaptive to have an alternative proof that adaptivity is needed for subadditive cost functions in the combinatorial cost problem that we study in this paper. 

    \paragraph{Pandora's Problem with order constraints.} An instance of Pandora's problem with order constraint is defined by $n$ boxes and a precedence graph $G$. Each box $b_i$ contains a random reward with distribution $D_i$ and has (additive) cost $c_i$. The precedence graph is a directed acyclic graph whose vertices are the boxes; to simplify the presentation we assume the existence of an auxiliary box $b_0$ (that has no cost and contains $0$ reward) that is the unique root of $G$. The decision maker selects (possibly in an adaptive way) which boxes to open sequentially, with the constraint that 
    in order to open a box $b_i$, at least one of the boxes corresponding to its parent nodes needs to have been already opened.
    %to open a box $b_i$ at least one of the boxes in its parent nodes needs to be already opened. 
    The goal, like in the standard Pandora's Problem, is to maximize the utility (i.e., maximum value minus sum of costs of opened boxes).
    
    \paragraph{Tree constraints.} 
    % To formally define $c_G$ we denote with $MST(S)$ the minimum spanning tree containing the root $b_0$ and the boxes in $S$, where the weight of a directed edge $e=(b_i,b_j)$ is $c_j.$
    % The cost function $c_G$ is defined as 
    % \[
    %     c_G(S) = \sum_{e = (b_i,b_j) \in MST(S)} c_j.
    % \]
    % The minimum spanning tree function is naturally subadditive by definition:
    % \begin{claim}
    %     For any precedence graph $G$, the corresponding cost function $c_G$ is subadditive.     
    % \end{claim}
    % 
    Given any rooted directed tree and subset of nodes $S$, we define $cl(S)$, the closure of $S$, as the minimal connected set of nodes containing $S$ and the root. Clearly the closure is monotone under inclusion, i.e., if $S\subseteq T$, then $cl(S) \subseteq cl(T)$. Alternatively, the closure of $S$ can be characterized as the union of all the paths from the root to the nodes in $S.$ We can use the closure operator to define a cost function $c_G$ as follows: $c_G(S) = \sum_{b_i \in cl(S)} c_i$.
    \begin{proposition}\label{prop:gs}
        If $G$ is a tree, then $c_G$ is gross substitutes.
    \end{proposition}
    \begin{proof}
    Let $S \subseteq T$ any two subsets of the nodes. Moreover, let $b$ be any element not in $ T$, to prove submodularity, we need to show that
    \[
        c_G(S \cup \{b\}) - c_G(S) \ge c_G(T \cup \{b\}) - c_G(T).
    \]
    If $b \in cl(T)$, then the inequality trivially holds, so we need to argue only about elements outside $cl(T)$ (and therefore also outside $cl(S)$).
    
    Since $b$ is not in $cl(S)$, there exists a unique path $s_0,s_1,\dots,s_k$ such that $s_0 \in S$, $s_i\notin S$ for all $i\in [k]$ and $s_k = b$, that connects $b$ to $S$. This means that $c_G(S \cup \{b\}) - c_G(S) = \sum_{i=1}^k c_i$. Now, since $b \not \in cl(T)$, there is also a path connecting $T$ to $b$. Since there is only one path from the root to $b$, we have that this second path is a suffix of $s_0,s_1,\dots,s_k$, starting at a certain $s_{\ell}\in T$. This means that $c_G(T \cup \{b\}) - c_G(T) = \sum_{i=\ell}^k c_i$ for some $0 < \ell \leq k-1$. This concludes the proof of the submodularity. 
    
    To argue about gross-substitutes, we use the equivalent definition of gross substitute we introduced at the end of \Cref{app:combinatorial}.
    We define $P_i, P_j$ and $P_k$ as the paths that connect $S$ (and hence its closure) to $i$, $j$ and $k$, respectively. 
    %First, note that we can restrict ourselves to consider only  cases when 
    Let $P_{i,j,k} = P_i \cap P_j \cap P_k $ and let $\hat S = S \cup P_{i,j,k}$. It holds that 
    \[
        \begin{cases}
            c_G(i, j|S) + c_G(k|S) = c_G(i, j|\hat S) + c_G(k|\hat S) + 2 \cdot c_G(P_{i,j,k}|S)\\
            c_G(i|S) + c_G(j,k|S) = c_G(i|\hat S) + c_G(j,k|\hat S) + 2 \cdot c_G(P_{i,j,k}|S)\\
            c_G(j|S) + c_G(i,k|S) = c_G(j|\hat S) + c_G(i,k|\hat S) + 2 \cdot c_G(P_{i,j,k}|S).\\
        \end{cases}
    \]
     
    Therefore, the multiset in Equation~\eqref{eq:gs} has a unique max if and only if the following has a non-unique maximum
    \[
        \{c_G(i, j|\hat S) + c_G(k|\hat S),c_G(i|\hat S) + c_G(j, k|\hat S), c_G(j|\hat S) + c_G(i, k|\hat S)\}.     
    \]
    If $P_i \cap P_j=P_i \cap P_k=P_j \cap P_k= P_{i,j,k}$ then the three terms share the same value:
    \[
        c_G(i, j|S) + c_G(k|S) = c_G(i|S) + c_G(j, k|S) = c_G(j|S) + c_G(i, k|S) = c_G(P_i|S) + c_G(P_j|S) + c_G(P_k|S).    
    \]
    Else, without loss of generality, $(P_i \cap P_j)\neq P_{i,j,k} $, and $P_k \cap (P_i \cup P_j) =  P_{i,j,k}$, so we get that 
    \[
        c_G(P_i \cup P_j|S) + c_G(P_k|S) < c_G(i|S) + c_G(j, k|S) = c_G(j|S) + c_G(i, k|S) = c_G(P_i|S) + c_G(P_j|S) + c_G(P_k|S).
    \]
    % \michal{I didn't understand the last component. Why is it wlog? can't it be that, say, $i$ is on the path to both $j$ and $k$, but none of $P_j$ and $P_k$ contained in each other?}
    % \federico{Write the last case: $P_j$ is a prefix of both.}
\end{proof}
    
    Fix any instance $\cI^G = (D_1,\dots,D_n, c_1, \dots, c_n, G)$ of Pandora's Problem with precedence constraint on $G$ and any instance $\cI = (D_1,\dots, D_n,c_G)$ of Pandora's problem with combinatorial cost function $c_G.$ 
    %We show that the two problems are somehow ``equivalent''. Formally, we have the following result.
    The following proposition establishes an equivalence relation between the two problems.
\begin{proposition} \label{cl:submodular-equivalence}
     For any optimal strategy $\pi_G^*$ for $\cI^G$, the strategy $\pi^*$ for $\cI$ that opens the same boxes in $\cI$ is optimal. 
     
     Moreover, 
     for any optimal strategy $\pi^*$ for $\cI$, the strategy $\pi_G^*$ for $\cI^G$ that  whenever $\pi^*$ opens box $i$, opens the boxes of $cl(i)$ that were not already opened, is optimal. 
\end{proposition}
Note that the transformations defined in \Cref{cl:submodular-equivalence} preserve the property of using a fixed-order.
 By \citet{BoodaghiansFLL20}, we know that there always exists a fixed order strategy for $\cI^G$ that is optimal and efficiently computable; this implies that there is a fixed order strategy for $\cI$ (same boxes and cost $c_G$) that is optimal for it.
    \begin{proof}
        Fix any instance $\cI^G$ of Pandora's problem with order constraint on a tree and the corresponding instance $\cI$ of Pandora's problem with the same boxes, but without order constraint and GS cost function $c_G$. Given any strategy $\pi^G$ for $\cI^G$, it is immediate to design a strategy $\pi$ for $\cI$ that attains at least the same expected utility: $\pi$ opens exactly the same boxes as $\pi^G$ (and thus earns exactly the same reward) and, given the structure of the cost function, it pays the same cost. 

    Conversely, fix any strategy $\pi$ for $\cI$,  we can associate a strategy $\pi^G$ for $\cI^G$ that yields at least the same utility. It opens the same boxes as $\pi$ but, before exploring a costly box $i$ (if it is not already opened) always exhausts the boxes within the unique path to box $i$ according to $G$.  All the boxes along this path have a $0$ marginal cost. Clearly, the latter strategy yields larger utility, as it gives larger rewards (as it opens more boxes) but pays at most the same cost (by submodularity).      
    \end{proof}

    \paragraph{General order constraint.} If the precedence graph $G$ is not a tree, we cannot prove an equivalence relation as above. 
    However, we can translate the (non-tree) example given in Theorem 6 of the extended version of \citet{BoodaghiansFLL20} into an instance of our model with a subadditive cost function, which admits no fixed-order strategy that is optimal.

    \begin{claim}
        \label{cl:subadditive-example}
        There exists an instance of Pandora's problem with subadditive cost $c$ that admits no fixed-order strategy that is optimal.
    \end{claim}
    \begin{proof}
        Consider the following instance on $4$ boxes. The rewards are weighted Bernoulli defined as follows:
        \[
            V_1=
            \begin{cases}
            100 &\text{ w.p. }1/3 \\
                2.5&\text{ w.p. }1/3\\
                0&\text{ w.p. }1/3  
            \end{cases}
            \quad \quad
            V_2 = 2
            \quad \quad
            V_3 = 
            \begin{cases}
                3&\text{ w.p. }1/2\\   
                0&\text{ w.p. }1/2
            \end{cases}
            \quad \quad
            V_4 = 
            \begin{cases}
                6&\text{ w.p. }1/2\\
                0&\text{ w.p. }1/2    
            \end{cases}
        \]
        The cost function is defined as follows: Box $1$ has a marginal cost of $0$ given any other subset of boxes. For the other boxes, the cost is given by: $c(\{2,3,4\}) = c(\{2,3\}) = 2.1$, $c(\{3\}) = c(\{3,4\}) = 1.1$, $c(\{2\}) = c(\{4\}) = c(\{2,4\})= 1$. 
        One can easily verify that this cost function is subadditive (but not submodular). This function is an adaptation of Theorem~9 of \citet{BoodaghiansFLL20} to our setting.
        
        %It is shown in \cite{BoodaghiansFLL20} that an optimal strategy in their model (which translates directly into our model) is the following: 
        One can verify that the optimal strategy is the following: 
        Open box $1$. 
        If $V_1=100$, halt. 
        Else, if $V_1=2.5$, open box $4$. 
        If $V_4=6$, halt. 
        Else ($V_4=0$), open box $3$ and halt. 
        Finally, if $V_1=0$, open box $4$. 
        If $V_4=6$, halt. 
        Else, open box $2$ and halt. 
        Moreover, one can also verify that no fixed-order strategy obtains the same expected utility or more.
        %\ffc{Check calculations.}
    \end{proof}

\section{Missing Proofs from Section \ref{sec:preliminaries}}
\label{app:preliminaries}

\OptimalFixThresholds*
\begin{proof}
%\bbc{Work in Progress.}
Let $\sigma:[n] \to [n]$ be a permutation,
and for sake of readability we assume without loss of generality that $\sigma$ is the identity permutation.
Consider any strategy $\pi$ which is restricted to inspecting the boxes in $[n]$ in the order $1,\ldots,n$.
Let $i\in [n]$, and let us assume that the strategy $\pi$ is in the stage where it has already inspected the boxes $\left[i-1\right]$, and it now has to decide whether to inspect box $i$ or halt.
Let $x\geq 0$ be the maximum value among those observed in the previous rounds.

Given these parameters,
we define $f_i(x)$ to be the maximum extra utility attainable by any strategy from this point onwards.
If $\pi$ decides to halt, then this extra utility is 0.
If $\pi$ decides to open box $i$, then the contribution of that box to the utility is $(V_{i} - x)^+ - c(i \mid [i-1])$,
and $\pi$ moves to round $i+1$, from which the maximum extra utility attainable is $f_{i+1}(\max(V_i,x))$.
We thus get the following backward-recursive relation:
$$
f_i(x) = \max\left[0, \E{(V_{i} - x)^+ - c(i \mid [i-1]) + f_{i+1}(\max(V_i,x))}\right]
$$
The base function $f_{n+1}$ is defined to be the constant function 0,
and the maximum utility attainable in total by any strategy with this order is $f_1(0)$.

Now, note that $f_i(x)$ is a monotone decreasing function of $x$.
Furthermore, $f_i$ is $1$-Lipshtitz and thus continuous. To see this, fix any two values $x<y$, and let $g_i(y)$ be the extra utility (with respect to $y$) attainable playing the same strategy underlying $f_i(x)$ (i.e., playing optimally pretending to have as largest reward found so far the value $x$ and not $y$). It is easy to see that $f_i(x) - g_i(y)$ is at most $y-x$, thus we have the following:
\[
    f_i(x) - f_i(y) \le f_i(x) - g_i(y) \le y - x,
\]
where the first inequality follows from the suboptimality of the strategy followed in $g_i$ given reward $y$. 

%\bbc{Prove that $f_i$ is continuous @Federico.}\ffc{Handwavy but of to me.}

Consider the set $C_i = \{x \in \mathbb{R}_{\geq 0} \mid f_i(x) = 0\}$.  For every $i$ we define
$$
t_i = 
\begin{cases}
  \infty  & C_i = \emptyset \\
  \min C_i & \text{Otherwise}
\end{cases}
$$

Since $f_i$ is continuous and monotonic, then $t_i$ is well defined.  Note also that $f_i(x) = 0 $ if and only if $x \geq t_i$.  Therefore,
%in order to attain $f_i(x)$ for every $i$ and $x$,
the strategy $\pi$ 
%is the strategy with thresholds $t_i$,
that for every round $i$
halts if and only if the maximum value observed thus far, $x$, is at least $t_i$,
is an optimal strategy.
This concludes the proof since the strategy we described is a fixed order strategy with thresholds $t_i$.
\end{proof}
\section{Missing Proofs from Section \ref{sec:MTT}}\label{app:MTT}
In this appendix we provide missing proofs from Section \ref{sec:MTT}.

\bottomEqualsTopMinusPpi*
\begin{proof}    
The first bullet is clearly implied by Observation \ref{obs:utility_expansion}, and the fact that for every $i \in [n]$ we have
$$
v_i - v_\A \leq (v_i - v_\A)^+ \leq v_i
$$
The second bullet holds since
\begin{align*}
    \util_N(\pi \mid T)  &= \E{\max_{i \in S(\pi)} (V_i)} - \E{c\left(S(\pi) \mid \{\A\} \cup T\right)} \\
                         &= p_{\left(\pi\right)} \cdot \E{\max_{i \in S(\pi)} (V_i) \mid \exists i \in \pi \textrm{ s.t. } V_i = v_i} 
                             - \E{c\left(S(\pi) \mid \{\A\} \cup T\right)} + p_{\left(\pi\right)}v_\A - p_{\left(\pi\right)}v_\A\\
                         &= p_{\left(\pi\right)} \cdot \E{\max_{i \in S(\pi)} (V_i - v_\A) \mid \exists i \in \pi \textrm{ s.t. } V_i = v_i} 
                             - \E{c\left(S(\pi) \mid \{\A\} \cup T\right)} + p_{\left(\pi\right)}v_\A \\
                         &= \util_M(\pi \mid T) +  p_{\left(\pi\right)}v_\A
\end{align*}

Finally, the third bullet is immediately implied by the fact that $\left(v_i-v_\A\right)^+ = v_i - v_\A$ for every $i \in \pi^Y$, and by the fact that 
\begin{align*}
    \util_Y(\pi \mid T) &= \E{\max_{i \in S(\pi)} (V_i - v_\A)^+} - \E{c\left(S(\pi) \mid \{\A\} \cup T\right)} \\
                        &= p_{\left(\pi\right)} \cdot \E{\max_{i \in S(\pi)} (V_i - v_\A)^+ \mid \exists i \in \pi \textrm{ s.t. } V_i = v_i} 
                            - \E{c\left(S(\pi) \mid \{\A\} \cup T\right)}
\end{align*}
\end{proof}

%\CancellationLemma*
    
\kGreaterThanOne*
\begin{proof}

    Assume towards contradiction that $k=0$.  In this case we have $\pi^N = \left(B^\pre,A^\suff\right)= \left(\pi^B,\pi^A\right)$.
    We split to cases: (a) $\pi^B$ is a permutation of $\pi^Y$, (b) $\pi^B$ is strictly included in $\pi^Y$. 
    
    \textbf{Case (a):}
    Note that we have $\util_N\left(\pi^N\right) = \util_N\left(\pi^B\right) + q_{\left(\pi^B\right)} \util_N\left(\pi^A \mid \pi^B\right)$.
    Furthermore, we have $p_{\left(\pi^Y\right)} = p_{\left(\pi^B\right)}$ since  $\pi^B$ and $\pi^Y$ are permutations of each other.
    Thus, since $\pi^N$ is optimal for the scenario that $V_\A = 0$, then in particular we must have $\util_N\left(\pi^B\right) \geq \util_N\left(\pi^Y\right)$ as otherwise we could replace $\pi^N$ by the strategy $\left(\pi^Y,\pi^A\right)$ and strictly improve utility --- note that the fact that $\pi^B$ and $\pi^Y$ are permutations of one another also implies $q_{\left(\pi^B\right)} \util_N\left(\pi^A \mid \pi^B\right) = q_{\left(\pi^Y\right)}\util_N\left(\pi^A \mid \pi^Y\right)$.
    
    Similarly to the proof of Lemma \ref{lem:easy_cases}, from here we get
    $$
    \util_N\left(\pi^Y\right) = \util_Y\left(\pi^Y\right) + p_{\left(\pi^Y\right)}v_\A \geq \util_Y\left(\pi^B\right) + p_{\left(\pi^B\right)}v_\A
    = \util_N\left(\pi^B\right) \geq \util_N\left(\pi^Y\right),
    $$
    where the first and second equalities hold by Observation \ref{obs:bottom_equals_top_minus_p_pi_v_pi}, 
    and the first inequality holds by the optimality of $\pi^Y$ for the scenario that $V_\A = v_\A$.
    Thus all expressions in the above chain are equal and in particular we have $\util_N\left(\pi^Y\right) = \util_N\left(\pi^B\right)$.
    This implies that the strategy $\pi'$ obtained from $\piStar$ by replacing $\pi^B$ with $\pi^Y$ is also optimal.
    Now consider the impulsive strategy $\pi'' = \left(\pi^Y, \A,\pi^A\right)$.
    Since $v_i \geq v_\A$ for any $i \in \pi^Y$,
    then the maximum value observed by $\pi''$ coincides with that of $\pi'$ for any realization of the boxes.
    On the other hand the cost incurred by $\pi''$ is weakly less then that of $\pi'$, again for any realization of the boxes.
    Thus the impulsive strategy $\pi''$ is optimal as well, a contradiction.

    \textbf{Case (b):}  Note that in this case, since $\pi^Y$ (and therefore also $\pi^B$) does not have a deterministic box,
    then in particular we have $q_{\left(\pi^B\right)} - q_{\left(\pi^Y\right)} > 0$.
    Now, the fact that $\pi^N = \left(\pi^B,\pi^A\right)$ implies
    \begin{align}
    \util_M(\pi^N_A \mid \pi^Y) = q_{\left(\pi^B\right)}\util_M(\pi^A \mid \pi^Y) \label{eq:k=0}
    \end{align}
    since in the impulsive strategy with dummy boxes $\left(\pi^B,\pi^A\right)_A$,
    the strategy $\pi^A$ is executed if and only if all the dummy boxes corresponding to $B$ are not realized, which happens with probability $q_{\left(\pi^B\right)}$.
    From here we get
    \begin{align*}
        0   &\leq \util_Y(\pi^Y) - \util_Y(\pi^B) \\
            &= \util_M(\pi^Y) - \util_M(\pi^B) \\
            &= \util_M(\pi^Y) - \util_M(\pi^N_B) \\
            &\leq \util_M(\pi^N_A \mid \pi^Y) - q_{\left(\pi^Y\right)}\util_M\left(\pi^A \mid \pi^Y\right) \\
            &= q_{\left(\pi^B\right)}\util_M(\pi^A \mid \pi^Y) - q_{\left(\pi^Y\right)}\util_M\left(\pi^A \mid \pi^Y\right) \\
            &= \left(q_{\left(\pi^B\right)} - q_{\left(\pi^Y\right)}\right) \util_M\left(\pi^A \mid \pi^Y\right),
    \end{align*}
    where the first inequality holds since $\pi^Y$ is optimal for the scenario where $V_\A = v_\A$,
    the first equality holds by Observation \ref{obs:bottom_equals_top_minus_p_pi_v_pi},
    the second equality holds since the impulsive strategy with dummy boxes $\left(\pi^B,\pi^A\right)_B$ exactly coincides with the strategy $\pi^B$,
    the second inequality holds by inequality (\ref{eq:low_branch_inequality}),
    and the third equality holds by equation (\ref{eq:k=0}).

    Since $q_{\left(\pi^B\right)} - q_{\left(\pi^Y\right)} > 0$, then this implies $\util_M\left(\pi^A \mid \pi^Y\right) \geq 0$, a contradiction
    to 
    %inequality~\eqref{eq:u_M(piA mid piY) < 0}. \tec{fix the reference to the inequality}
    \Cref{clm:u_M(piA mid piY) < 0}.
    %\bbc{Is this clear?}
\end{proof}

%\lowBranch*

\section{Missing Proofs from Section~\ref{sec:reduction}} \label{app:reduction}

In this appendix, we show that the discretization parameter $\kappa$ is well defined and that, similarly to what is done in Proposition~\ref{prop:maintain-sm}, the reduction $\transformber$ preserves several additional properties.

    \begin{lemma}
    \label{lem:kappa}
        Let $V$ be a non-negative random variable with finite expectation, then the following holds: $\lim_{\kappa \to \infty}\E{(V - \kappa)^+} = 0$.
    \end{lemma}
    \begin{proof}
        We use the following equality that holds for any non-negative random variable:
        \[
            \E{V} = \int_0^{+\infty} \P{V \ge \lambda} \text{d}\lambda,
        \]
        where the integral is the standard Lebesgue integral.
        In particular, the previous equality implies that 
        \begin{equation}
        \label{eq:tail}
            \lim_{\lambda \to + \infty} \P{V \ge \lambda} = 0.
        \end{equation}
        For any positive integer $\kappa$, we have the following:
        \begin{align*}
             \E{(V-\kappa)^+} &= \int_0^{+\infty} \P{(V-\kappa)^+ \ge \lambda} \text{d}\lambda = \int_0^{+\infty} \P{V \ge \lambda + \kappa} \text{d}\lambda,
        \end{align*}
        where we used the fact that the integrand is upper bounded by $1$ and $\{\lambda = 0\}$ has $0$ Lebesgue measure.
        We want to use Lebesgue's dominated convergence theorem (see, e.g. Theorem $11.32$ of \citet{Rudin76}) to argue that the last term of the previous equality goes to $0$ as $\kappa$ goes to infinity. 
        
        To this end, define the sequence of functions $f_{\kappa}: [0,\infty) \to [0,1]$ as $f_\kappa(\lambda) = \P{V \ge \lambda + \kappa}$ and the function  $g:[0,\infty) \to [0,1]$ as $g(\lambda) = \P{V \ge \lambda}$.  We verify the hypothesis of Lebesgue's dominated convergence theorem:
        $(i)$ $0 \le f_{\kappa}(\lambda) \le g(\lambda)$ for all $\lambda\in[0,+\infty)$ and all $\kappa$ (monotonicity of probability), $(ii)$ $g$ is integrable in $[0,+\infty)$ as $\int_0^{+\infty} g(\lambda) \text{d}\lambda$ is exactly $\E{V}$, $(iii)$ for any fixed $\lambda$ it holds that $\lim_{\kappa \to +\infty} f_{\kappa}(\lambda) = 0$ (stated differently, $\lim_{k \to +\infty} \P{V \ge k + \lambda} = 0$ for all fixed $\lambda$ because $\E{V} < \infty$, see \Cref{eq:tail}). All in all, we have that
        \[
            \lim_{\kappa \to +\infty}\E{(V-\kappa)^+} = \lim_{\kappa \to +\infty} \int_0^{+\infty} f_{\kappa}(\lambda) \text{d}\lambda = \int_0^{+\infty} \lim_{\kappa \to +\infty} f_{\kappa}(\lambda) \text{d}\lambda = 0.
        \]
    \end{proof}

\begin{claim}\label{cl:mrf}
If $c$ is a matroid rank function (MRF), then $c'$ obtained by transformation $\transformber$ is also MRF. 
\end{claim}
\begin{proof}
    Cost function $c$ is MRF, thus there exists a matroid $\mathcal{M}\subseteq 2^{[n]}$ such that for every $S \subseteq [n]$, $c(S) =r_{\mathcal{M}}(S)$.
For a set $S\subseteq [n] \times [m]$, let $A_S= \{i \mid \exists j \mbox{ such that } (i,j)\in S\}$, and let $B^i_S = |\{ j \mid  (i,j)\in S \}|$.
Consider the family of subsets $\mathcal{M}' \subseteq 2^{[n]\times[m]}$, where 
\[
    \mathcal{M}' =\{S\mid A_S\in \mathcal{M}, B_S^i \leq 1 \,  \forall{i\in[n]}\}.
\]
We claim that $\mathcal{M}'$ is a matroid. It is easy to verify that $\mathcal{M}'$ is downward-closed, so we only need to verify the Augmentation Property (see Appendix~\ref{app:combinatorial}).
Let $S,T\in \mathcal{M}'$, such that $|S|>|T|$. 
$\mathcal{M}$ is a matroid, thus, since $A_S,A_T\in\mathcal{M}$, and $|A_S|>|A_T|$, there exists $i \in A_S \setminus A_T$ such that $A_T\cup\{i\}\in \mathcal{M}$. This implies that there exists $j$ such that $ (i,j)\in S$. Therefore, since $(i)$ $A_{T\cup\{(i,j)\}}= A_T \cup \{i\}$, $(ii)$ $B^{i'}_{T\cup\{(i,j)\}} =B^{i'}_{T}$ for $i'\neq i$, and $(iii)$  $B^{i}_{T\cup\{(i,j)\}}=1$, it holds that $T\cup\{(i,j)\}\in \mathcal{M}'$. The function $c'$ is the rank function of the matorid $\mathcal{M}'$.
\end{proof}
\begin{claim}\label{cl:gs}
If $c$ is gross-substitutes, then $c'$ obtained by transformation $\transformber$ is also gross-substitutes. 
\end{claim}
\begin{proof}

Given $T \subseteq [n] \times [m]$, let $A_T =\{i \mid \exists j, (i,j)\in T\}$.

To prove the claim, we first prove the following two observations. Given a vector of prices $p$ for $[n]\times [m]$, let $q$ be the vector of prices for $[n]$ such that $q_i=\min_j p_{i,j}$. As a convention, we denote with $p(S)$ the sum of all the prices of the elements in $S$.
 Then 
\begin{equation} \label{eq:1}
     S\in \arg\max_{T \subseteq [n] \times [m]} c'(T) -p(T) \implies A_S \in \arg\max_{T \subseteq [n]} c(T)-q(T) 
\end{equation}
\begin{equation} \label{eq:2}
       q(A_S) =p(S) \wedge A_S \in \arg\max_{T \subseteq [n]} c(T)-q(T)  \implies 
      S\in \arg\max_{T \subseteq [n] \times [m]} c'(T) -p(T)
\end{equation}
We first observe that: $(i)$ for every $S\subseteq [n]\times [m]$, it holds that $c'(S)=c(A_S)$, and $q(A_S)\leq p(S)$. $(ii)$ For each set $A\subseteq [n]$, then for the set $T=\{(i,\arg\min_{j\in[m]} p_{i,j}) \mid i\in A\}$, it holds that $p(T)=q(A)$.

Equation~\eqref{eq:1} follows since  given a set $S\in\arg\max_{T \subseteq [n] \times [m]} c'(T) -p(T)$, assume towards contradiction that there exists $A\subseteq [n]$ such that $c(A)-q(A)> c(A_S)-q(A_S)$, then if we let $T=\{(i,\arg\min_{j\in[m]} p_{i,j}) \mid i\in A\}$, it holds that $p(T)=q(A)$, then $c'(T)-p(T)= c(A)-q(A) > c(A_S)-q(A_S) \geq c'(S)-p(S)$ which is a contradiction.

Equation~\eqref{eq:2} follows since given $S\subseteq [n]\times[m]$ that satisfies the conditions of the equation, for every $S' \subseteq [n]\times [m]$, it holds that $$ c'(S')-p(S') \leq c(A_{S'})-q(A_{S'}) \leq c(A_S)-q(A_S) =c'(S)-p(S).$$

To conclude the proof, we need to argue that for any vectors of prices $p,p'$ over $[n]\times [m]$, such that $p'\geq p$ component-wise,  and for every set $S\in \arg\max_T c'(T) -p(T)$ there exists $S'\in \arg\max_T c'(T)-p'(T)$  that contains  $R=\{(i,j)\in S \mid p_{i,j}=p'_{i,j}\}$.
It holds that $p(S) =q(A_S)$ (otherwise $T=\{(i,\arg\min_{j\in[m]} p_{i,j}) \mid i\in A_s\}$ has the same value with lower price).
Let $q'_i =\min_j p'_{i,j}$.
Since for every $i\in A_R$, it holds that $q_i=q'_i$, and 
by gross-substitutes of $c$, it holds that there exists $A \in \arg\max_{T\subseteq [n]} c(T)-q'(T)$, that contains $A_R$. 
It holds that $q'(A_R) =p'(R) =p(R)$ since otherwise, there exists $R'$ that $A_{R'}=A_R$, and  $p(R')< p(R)$ contradicting the optimality of $S$.
Thus, for $S'=R\cup \{(i,\arg\min_{j\in[m]} p_{i,j}) \mid i\in A \setminus A_R\}$, it holds that $q'(A_{S'})=p(S')$ and $A_{S'}\in \arg\max c(T)-q'(T)$, then by Equation~\eqref{eq:2}, we found $S'$ that is in the demand and contain $R$, which concludes the proof.
\end{proof}

\begin{claim}\label{cl:coverage}
If $c$ is coverage, then $c'$ obtained by transformation $\transformber$ is also coverage. 
\end{claim}
\begin{proof}
    Let $E$, $w$, and $g$ be elements and the mappings that are defined in the definition of coverage with respect to $c$. We construct $g':E \rightarrow 2^{[n]\times[m]}$ such that $E$, $w$, and $g'$ define $c'$ as a coverage.
    In particular, we have 
    \[
        g'(e) = g(e) \times [m], \, \forall e\in E. \qedhere
    \]
\end{proof}

\begin{claim}\label{cl:xos}
If $c$ is XOS, then $c'$ obtained by transformation $\transformber$ is also XOS. 
\end{claim}
\begin{proof}
    Cost function $c$ is XOS on $[n]$, this means that there exist additive functions $a^1,\ldots,a^\ell$ over $[n]$, such that for every $S \subseteq [n]$, $c(S)=\max_{t\in [\ell]} a^t(S)$. We want to show that $c'$ is XOS over $[n] \times [m]$. To this end we construct the following family of additive functions over $[n] \times [m]$: for every $t\in [\ell]$, and $r\in [m]^n$, define the additive function $a^{t,r}(i,j)= a^t(i) \cdot \indicator{r_i=j}$. We conclude the proof by observing that 
    \[
        c'(S) = \max_{t\in[\ell], r\in [m]^n}a^{t,r}(S), \, \forall S\subseteq [n]\times [m]. \qedhere
    \]
\end{proof}
\begin{claim}\label{cl:subadditive}
If $c$ is subadditive, then $c'$ obtained by transformation $\transformber$ is also subadditive. 
\end{claim}
\begin{proof}
We need to prove that for any pair of sets $S,T \subseteq [n] \times [m]$, it holds that $c'( S) + c'( T) \geq c'(S\cup T) $.
Let $A_S= \{i \mid \exists j\in [m]\mbox{ such that } (i,j)\in S \}$,  $A_T= \{i \mid \exists j\in [m]\mbox{ such that } (i,j)\in T \}$, and $A_{S\cup T} = \{i \mid \exists j\in [m]\mbox{ such that } (i,j)\in S\cup T \}$.
Then it holds that 
\[
    c'( S) + c'(T) = c(A_S)+c(A_T) \geq c(A_{S\cup T}) =c'(S\cup T),
\]
where the inequality is by subadditivity of $c$ and since $A_S \cup A_T = A_{S\cup T}$.
\end{proof}
\begin{claim}\label{cl:budget}
Transformation $\transformber$ does not maintain budget additive. 
\end{claim}
\begin{proof}
Consider the case where $n=3,m=2$ and $c(S) = \min(|S|,2)$ (a symmetric function with 3 values that equal $1$ and a budget of $2$, and we create two copies of every original box).
$c'$ must define a value of $1$ for all boxes, and a budget of $2$, but then $c'(\{(1,1),(1,2)\})=2$, which should be $1$ by the definition of the transformation.
\end{proof}

\end{document}